\documentclass[11pt,a4paper,english]{amsart}
\usepackage[utf8]{inputenc}
\usepackage[english]{babel} 

\usepackage{amsmath} 
\usepackage{amsthm} 
\usepackage{graphicx} 
\usepackage{xcolor} 
\usepackage{amsfonts}
\usepackage[biblabel]{cite}
\usepackage{bm} 
\usepackage[hidelinks]{hyperref} 
\usepackage{amssymb} 
\usepackage{tikz-cd} 
\usepackage{amscd}
\usepackage{etoolbox}
\patchcmd{\thmhead}{(#3)}{#3}{}{}
\usepackage{enumitem}

\DeclareMathOperator{\ini}{in} 
\DeclareMathOperator{\wt}{wt}
\DeclareMathOperator{\ev}{ev} 
\DeclareMathOperator{\Tr}{Tr}
\DeclareMathOperator{\Trr}{Tr_{\ell}}
\DeclareMathOperator{\RS}{RS}
\DeclareMathOperator{\PRS}{PRS}

\newcommand{\F}{{\mathbb{F}}}
\newcommand{\fq}{\mathbb{F}_q}
\newcommand{\fqs}{\mathbb{F}_{q^s}}

\newcommand{\Pp}{{\mathbb{P}^{1}}}
\newcommand{\X}{{\mathbb{X}}}
\newcommand{\Z}{{\mathbb{Z}}}
\newcommand{\A}{{\mathbb{A}}}
\newcommand{\B}{{\mathcal{B}}}
\newcommand{\II}{{\mathfrak{I}}}
\newcommand{\E}{{\RS(N,\Delta)}}
\newcommand{\ET}{{(\RS(N,\Delta)_q)^\perp}}
\newcommand{\EE}{{\RS(N,\Delta)_q}}
\newcommand{\DT}{{(\PRS(N,\Delta)_q)^\perp}}
\newcommand{\mDT}{{(\mathcal{D}(N,\Delta)^\perp)_q}}
\newcommand{\mDTq}{{(\mathcal{D}(N,\Delta)^\perp)_{q^2}}}
\newcommand{\mDTT}{{((\mathcal{D}(N,\Delta)^\perp)_q)^\perp}}
\newcommand{\mDTh}{{(\mathcal{D}(N,\Delta)^\perp)_{q^2}}}
\newcommand{\mDTTh}{{((\mathcal{D}(N,\Delta)^\perp)_{q^2})^{\perp_h}}}
\newcommand{\Dtt}{{\PRS(N,\Delta)^\perp}}
\newcommand{\D}{{\PRS(N,\Delta)}}
\newcommand{\mD}{\mathcal{D}(N,\Delta)}
\newcommand{\DS}{{\PRS(N,\Delta)_q}}

\newcommand{\Dr}{{\mathcal{D}(\Trr,\Delta)}}

\newcommand{\Drt}{{\mathcal{D}(\Trr,\Delta)^\perp}}
\newcommand{\Drts}{{(\mathcal{D}(\Trr,\Delta)^\perp)_{q^2}}}
\newcommand{\Drtts}{{((\mathcal{D}(\Trr,\Delta)^\perp)_{q^2})^{\perp_h}}}

\newcommand{\Drtshn}{{(\mathcal{D}(\Trr,\Delta)^\perp)_{q^2}}}

\newcommand{\Xr}{{\mathbb{X}_{\Trr}}}
\newcommand{\xr}{{X_{\Trr}}}
\newcommand{\T}{{\mathcal{T}}}

\usepackage{mathtools} 
\DeclarePairedDelimiter\abs{\lvert}{\rvert}%
\DeclarePairedDelimiter\norm{\lVert}{\rVert}%

\makeatletter
\let\oldabs\abs
\def\abs{\@ifstar{\oldabs}{\oldabs*}}
\let\oldnorm\norm
\def\norm{\@ifstar{\oldnorm}{\oldnorm*}}
\makeatother

\newtheorem{thm}{Theorem}[section]
\newtheorem{prop}[thm]{Proposition}
\newtheorem{cor}[thm]{Corollary}
\newtheorem{lem}[thm]{Lemma}
\theoremstyle{definition}
\newtheorem{defn}[thm]{Definition} 

\newtheorem{rem}[thm]{Remark} 
 
\newtheorem{ex}[thm]{Example}

\title[EAQECCs from subfield subcodes of projective Reed-Solomon codes]{Entanglement-assisted quantum error-correcting codes from subfield subcodes of projective Reed-Solomon codes}
\author{Philippe Gimenez, Diego Ruano and Rodrigo San-José}
\curraddr{
\texttt{Philippe Gimenez, Diego Ruano, Rodrigo San-José:} IMUVA-Mathematics Research Institute, Universidad de Valladolid, 47011 Valladolid (Spain).
}
\email{pgimenez@uva.es;  diego.ruano@uva.es; rodrigo.san-jose@uva.es}
\date{}
\thanks{This work was supported in part by the following grants: Grant TED2021-130358B-I00 funded by MCIN/AEI/10.13039/501100011033 and by the ``European Union NextGenerationEU/PRTR'', Grants PID2022-138906NB-C21 and PID2022-137283NB-C22 funded by MCIN/AEI/10.13039/501100011033 and by ERDF ``A way of making Europe'', FPU20/01311 funded by the Spanish Ministry of Universities, and by QCAYLE project funded by MCIN, the European Union NextGenerationEU (PRTR C17.I1) and Junta de Castilla y Le\'on.}

\subjclass[2020]{Primary: 81P70. Secondary: 94B05, 13P25}

\keywords{asymmetric quantum codes, EAQECC, evaluation codes, linear codes, projective Reed-Solomon codes, subfield subcodes, trace}

\begin{document}

\maketitle

\begin{abstract}
We study the subfield subcodes of projective Reed-Solomon codes and their duals: we provide bases for these codes and estimate their parameters. With this knowledge, we can construct symmetric and asymmetric entanglement-assisted quantum error-correcting codes, which in many cases have new or better parameters than the ones available in the literature.
\end{abstract}

\section{Introduction}

The subfield subcode of a linear code $C\subset \F_{q^s}^n$, with $s\geq 1$, is the linear code $C\cap \F_q^n$. Considering subfield subcodes is a standard technique for constructing long linear codes over a small finite field. For instance, BCH codes are obtained in this way. They can be regarded as subfield subcodes of Reed-Solomon codes and their duals \cite{bierbrauercyclic}. In this work, we study subfield subcodes of projective Reed-Solomon codes.

Reed-Solomon codes are constructed by evaluating one-variable polynomials at points of the affine line. They have optimal parameters, although they cannot be defined over a small finite field. Projective Reed-Solomon codes are constructed by evaluating two-variable homogeneous polynomials at points of the projective line. When one evaluates at all the points they are commonly called doubly extended Reed-Solomon codes. Subfield subcodes of projective Reed-Solomon codes, when one evaluates at all the points of the projective line, were studied in \cite{bierbrauerrs}. 

In this work we consider a more general setting: we may evaluate at fewer points to define a projective Reed-Solomon code and then compute its subfield subcode. We provide bases for both the subfield subcodes of projective Reed-Solomon codes and their duals and, thus, a formula for their dimension. For the dual code, we use Delsarte's Theorem \ref{delsarte}, for which we need to study first the metric structure of the codes we are considering. We also study the vanishing ideal of the points in which we evaluate, which allows us to discuss linear independence between the traces that arise when using Delsarte's Theorem. Moreover, we estimate the minimum distance for both primary and dual codes. For the primary code we simply use the bound given by the projective Reed-Solomon code, and for the dual one we use a BCH-type bound.

Reed-Solomon and BCH codes have been extensively used to construct quantum codes using the CSS construction, see for instance \cite{bierbrauerquantum,laguardialibro,galindoBCH}. It is therefore natural to consider subfield subcodes of projective Reed-Solomon for constructing quantum codes. 

The construction of quantum computers has important consequences because of their computing capabilities. Despite the fact that quantum mechanical systems are sensitive to disturbances and arbitrary  quantum states cannot be replicated, error correction is possible. Quantum error-correcting codes are designed for protecting quantum information from quantum noise and particularly decoherence. An important class of quantum error-correcting codes are stabilizer codes; they can be derived from classical ones by using self-orthogonal codes for the symplectic product \cite{calderbankgoodquantum}. One can also consider the Euclidean and the Hermitian inner product, and we will call the resulting quantum error-correcting codes QECCs. Entanglement-assisted quantum error-correcting codes (EAQECCs) constitute an extension of quantum codes. EAQECCs make use of pre-existing entanglement between transmitter and receiver to correct more errors \cite{entanglement,galindoentanglement}. One virtue of this class of codes is that one can get a quantum code from any linear code without any assumption on dual containment. The main additional task for EAQECCs is to give formulae to obtain the optimal number $c$ of maximally entangled pairs of qudits needed.

Moreover, both for QECCs and EAQECCs one can consider the asymmetric case \cite{ioffe,sarvepalliasymmetric,galindoasymmetric}. Asymmetric quantum codes have a different error-correction capability for phase-shift and qudit-flip errors. These two types of errors are not equally likely, and it is desirable to construct quantum codes with a higher correction capability for phase-shift errors \cite{ioffe}.

In this work, we provide EAQECCs with excellent parameters coming from different constructions. In the Euclidean case, we are able to obtain both symmetric and asymmetric EAQECCs with excellent parameters from subfield subcodes of projective Reed-Solomon codes. A key fact for the construction of these codes and the computation of their parameters is the knowledge of the parameters and structure of both the primary and dual codes. We also obtain QECCs, i.e. EAQECCs without entanglement assistance, from subfield subcodes of projective Reed-Solomon codes in some cases. By considering the Hermitian inner product we are also able to obtain codes with excellent parameters. In fact, we produce new parameters according to \cite{codetables}. Furthermore, as we are giving several different constructions using subfield subcodes of projective Reed-Solomon codes, this contributes to expanding the known constellation of parameters for EAQECC.

Finally, we consider the codes in \cite{fernandotrace}, Reed-Solomon, and BCH codes obtained by evaluating at the roots of a trace function. We consider the projective version of the codes in \cite{fernandotrace}, that is, the subfield subcodes of projective Reed-Solomon codes evaluating at the roots of a trace function and the point at infinity. This allows us to give classical linear codes which are record in \cite{codetables}, and new EAQECCs. 

Our results can be summarized as follows.
\begin{itemize}[noitemsep,topsep=0pt]
    \item We consider projective Reed-Solomon codes over the zero locus of $x^N-x$ (and the point at infinity), where we evaluate an arbitrary set of monomials. We obtain bases for the subfield subcodes of these codes in Theorem \ref{baseproyectiva}.
    \item When $p\mid N$, bases for the duals of the subfield subcodes are obtained in Theorem \ref{baseproyectivadual}.
    \item Considering sets of monomials whose exponents are a union of consecutive cyclotomic sets and the next minimal element, we obtain EAQECCs with entanglement parameter $c\leq 1$ in Theorem \ref{cuanticoeuclideo} and Theorem \ref{cuanticohermitico}. Some of the resulting codes improve the table for EAQECCs from \cite{codetables}.
    \item Assuming $p\mid N$, by considering the sets of monomials $\{0,1,\dots,d_i\}$, for some $1\leq d_1,d_2\leq N-1$, we obtain asymmetric EAQECCs with entanglement parameter $c=1$, which compare favorably with the current literature.
    \item By evaluating in the zeroes of the trace function, plus the point at infinity, and evaluating monomials whose exponents are a union of consecutive cyclotomic sets and the next minimal element, we obtain linear codes with good parameters in Theorem \ref{thmtraza}, some of which improve the best known parameters in \cite{codetables}, see Example \ref{ejemplocerostraza}. Moreover, we obtain EAQECCs with good parameters and entanglement parameter $c\leq 1$ in Theorem \ref{paramquantraza}.
\end{itemize}

\section{Preliminaries}

We consider a finite field $\fq$ of $q$ elements with characteristic $p$, and its degree $s$ extension $\F_{q^s}$, with $s\geq 1$. We consider the affine space $\mathbb{A}^1$ over $\F_{q^s}$ and the polynomial ring $R=\fqs[x]$. We choose a set of elements $Y=\{Q_1,\dots,Q_n\}\subset \mathbb{A}^1$ and its vanishing ideal $I(Y)=\langle \prod_{i=1}^n (x-Q_i)\rangle $, where we are regarding the points of $\mathbb{A}^1$ as elements in $\F_{q^s}$. We define the following evaluation map
$$
\ev_Y:R/I(Y) \rightarrow \fqs^{n},\:\: f\mapsto \left(f(Q_1),\dots,f(Q_n)\right)_{Q_i \in Y}.
$$
where we denote a polynomial and its class in the quotient ring $R/I(Y)$ in the same way. Let $\Delta$ be a subset of $\{0,1,\dots,n-1\}$. Then, the Reed-Solomon code associated to $\Delta$ and $Y$, denoted by $\RS(Y,\Delta)$, is the code generated by 
$$
\{\ev_Y(x^i)\mid i \in \Delta \}.
$$
The usual choices are $\Delta=\{0,1\dots,d\}$ and $Y=\F_{q^s}^*=\F_{q^s}\setminus \{0\}$, which give a Reed-Solomon code with parameters $[q^s-1,d+1,q^s-d-1]$. This code can be extended by evaluating at $0$ as well, obtaining a code with parameters $[q^s,d+1,q^s-d]$.

Let $N>1$ be such that $N-1\mid q^s-1$. We can consider the set of points $Y_N^*=\{Q_1,\dots,Q_N\}$ given by the zero locus of $I(Y_N^*)=\langle x^{N-1}-1\rangle$. In this case,  $Y_N^*$ forms a multiplicative subgroup of $\F_{q^s}^*$ and it is already known how to obtain bases for its subfield subcodes (see, for example, \cite{hattori,ssctoric}). Moreover, BCH codes can be defined as the duals of the subfield subcodes of Reed-Solomon codes when we evaluate in a subgroup $Y_N^*$ \cite{bierbrauercyclic}. Indeed, let $\alpha\in \F_{q^s}$ be a primitive $(N-1)$th root of unity. $C$ is a BCH code of designed distance $\delta$ if it has as generator polynomial the least common multiple of the minimal polynomials of the $\delta-1$ consecutive elements $\alpha^b,\alpha^{b+1},\dots,\alpha^{b+\delta-2}$, with $b\geq 1$, which implies that $C$ is formed by the vectors over $\F_q^{N-1}$ that are orthogonal to the rows of the matrix 
\begin{equation}\label{pseudoparity}
H=\begin{pmatrix}
1&\alpha^b&\alpha^{2b}&\cdots &\alpha^{(N-2)b}\\
1&\alpha^{b+1}&\alpha^{2(b+1)}&\cdots &\alpha^{(N-2)(b+1)}\\
\vdots & \vdots&\vdots & \ddots & \vdots \\
1& \alpha^{b+\delta-2} & \alpha^{2(b+\delta-2)}&\cdots & \alpha^{(N-2)(b+\delta-2)}
\end{pmatrix}.
\end{equation}
However, this is precisely the generator matrix of the Reed Solomon code over $\fqs$ with $\Delta=\{b,b+1,\dots,b+\delta-2 \}$ and $Y=Y_N^*$. Furthermore, the vectors in $\F_q^{N-1}$ that are orthogonal to the rows of $H$ are precisely the vectors of the subfield subcode of the dual code of this Reed-Solomon code, which is therefore equal to the aforementioned BCH code. In this situation, we say that $H$ is a \textit{pseudo parity check-matrix} for $C$.

Because of the previous discussion, throughout this work we will focus on evaluating in subgroups of the form $Y^*_N$ unless stated otherwise. As before, we can also include the evaluation of $0$, which corresponds to considering instead the set $Y_N$, the zero locus of $I(Y_N)=\langle x^{N}-x\rangle$. For the Reed-Solomon codes obtained by evaluating the associated monomials to $\Delta$ in $Y_N$ we will use the notation $\RS(N,\Delta)$. The subfield subcode of the code $\RS(N,\Delta)$ over $\fq$ is denoted by $\RS(N,\Delta)_q:= \RS(N,\Delta) \cap \F_q^{N}$. In this case, for the sake of simplicity, we are also going to denote $R_N:=R/I(Y_N)$.

Now we are going to introduce some necessary definitions in order to obtain bases for the codes $\RS(N,\Delta)_q$. We define $\mathbb{Z}_N=\{0\} \cup \mathbb{Z}/\langle N-1 \rangle$, where we represent the classes of $\mathbb{Z}/\langle N-1 \rangle$ by $\{1,\dots,N\}$. A subset $\mathfrak{I}$ of $\Z_N$ is called a \textit{cyclotomic set} with respect to $q$ if $q\cdot z \in \II$ for any $z\in \II$. $\II$ is said to be minimal (with respect to $q$) if it can be expressed as $\II=\{q^i\cdot z,i=1,2,\dots\}$ for a fixed $z\in \II$, and in that situation we will write $\II_z:=\II$ and $n_z=\abs{\II_z}$. We say $z$ is a \textit{minimal representative} of $\II_z$ if $z$ is the least element in $\II_z$, and we will say it is a \textit{maximal representative} of $\II_z$ if it is the biggest element. We will denote by $\mathcal{A}$ the set of minimal representatives of the minimal cyclotomic cosets, and by $\mathcal{B}$ the set of maximal representatives of the minimal cyclotomic cosets. 

\begin{ex}\label{excyclo}
Consider the extension $\F_9\supset \F_3$. We consider $N=9$ and we have $\Z_N=\{0\}\cup \Z/\langle 8 \rangle$. We have the following minimal cyclotomic sets:
$$
\II_0=\{0\},\II_1=\{1,3\},\II_2=\{2,6\},\II_4=\{4\},\II_5=\{5,7\},\II_8=\{8\}.
$$
The set of minimal representatives is $\mathcal{A}=\{ 0,1,2,4,5,8 \}$, and the set of maximal representatives is $\B=\{0,3,4,6,7,8\}$.
\end{ex}
 
The dimension of the subfield subcodes of Reed-Solomon codes is already present in \cite{hattori}. For the codes $\RS(N,\Delta)_q$ it is possible to obtain a basis given by the evaluation of some polynomials. For each $a\in \mathcal{A}$, we define the following trace map:
$$
\mathcal{T}_a:R_N\rightarrow R_N,\:\: f\mapsto f+f^q+\cdots + f^{q^{(n_a-1)}},
$$
and given $\Delta\subset \{0,1,\dots,N-1\}$, we denote $\Delta_\II:=\bigcup_{\II_a\subset \Delta}\II_a\subset \Delta$. The following result gives a basis for the code $\RS(N,\Delta)_q$ \cite[Thm. 11]{galindolcd}.

\begin{thm}\label{baseafin}
Let $\Delta$ be a subset of $\{0,1,\dots,N-1\}$ and set $\xi_a$ a primitive element of the field $\F_{q^{n_a}}$. Then, a basis of the vector space $\RS(N,\Delta)_q$ is given by the images under the map $\ev_{Y_N}$ of the set of classes in $R_N$
$$
\bigcup_{a\in \mathcal{A}\mid \II_a\subset \Delta}\{\mathcal{T}_a(\xi_a^r x^a)\mid 0\leq r\leq n_a-1 \}.
$$
\end{thm}

As a consequence, we have that 
$$
\dim \RS(N,\Delta)_q=\sum_{\II_z:\II_z\subset \Delta}n_z=\abs{\Delta_{\II}}.
$$

Having seen the affine setting, we are now going to introduce the codes we are going to use throughout this work. We consider the projective line $\Pp$ over $\F_{q^s}$ and the polynomial ring $S=\fqs [x_0,x_1]$. Given a degree $d\geq 1$, we denote by $S_d$ the homogeneous polynomials of degree $d$. We are going to fix representatives for the points of $\Pp$ in the following way: for each point $[P]\in \Pp$, we choose the representative whose first nonzero coordinate is equal to 1. We will denote by $P^1$ this set of representatives, seen as points in the affine space $\mathbb{A}^{2}$, and we will call them \textit{standard representatives}. If we also consider a finite set of points $X=\{Q_1,\dots,Q_n\}\subset P^1$, we can define the following evaluation map
$$
\ev_X:S/I(X) \rightarrow \fqs^{n},\:\: f\mapsto \left(f(Q_1),\dots,f(Q_n)\right)_{Q_i \in X},
$$
where, as before, we denote a polynomial in $S$ and its class in $S/I(X)$ in the same way. Given $\Delta\subset \{0,1,\dots,n-1\}$, we define $d(\Delta):=\max\{i\mid i \in \Delta\}$. The \textit{projective Reed-Solomon} code associated to $\Delta$ and $X$ is the code generated by 
$$
\{\ev_X(x_0^{d(\Delta)-i}x_1^i) \mid i \in \Delta \},
$$
which will be denoted by $\PRS(X,\Delta)$. We note that we are only evaluating monomials of exactly degree $d(\Delta)$, which means that their linear combinations are homogeneous polynomials of degree $d(\Delta)$. If $0\not \in \Delta$, $\PRS(X,\Delta)$ is a degenerate code because all the previous monomials would evaluate to $0$ at the point $[1:0]$. Therefore, we are always going to assume in what follows that $0\in \Delta$. Some authors define these codes over the projective space
without fixing representatives, as in \cite{martinez}, but then they can only define the code up to monomial equivalence. Monomially equivalent codes do not necessarily have monomially equivalent subfield subcodes, for example in \cite{fernandoGRS} the authors see that the dimension of the subfield subcode of a generalised Reed-Solomon code depends on the twist vector chosen, and that is why we fix representatives from the beginning.

Given a degree $1\leq d \leq q^s$, the most standard definition of projective Reed-Solomon code in the literature is the code $\PRS(P^1,\Delta_d)$, where $\Delta_d:=\{0,1,\dots, d\}$. The code $\PRS(P^1,\Delta_d)$ is also called \textit{doubly extended Reed-Solomon code} and its parameters are $[q^s+1,d+1,q^s-d+1]$. 

In order to obtain bases for the subfield subcodes of the codes $\PRS(X,\Delta)$, we are going to evaluate in subgroups similarly to the affine case. The natural ideal is to add the point at infinity $[0:1]$ to the points that we were considering in the affine case. Therefore, given $N$ such that $N-1\mid q^s-1$, we define $\X_N^*=[\{1\}\times Y_N^*]\cup [0:1]\subset \Pp$ and $\X_N=[\{1\}\times Y_N]\cup [0:1]\subset \Pp$, where we recall that $Y_N^*$ and $Y_N$ are the zero locus of $\langle x^{N-1}-1\rangle$ and $\langle x^{N}-x\rangle$, respectively. However, it is easy to see that another set of representatives for $\X^*_N$ is $[Y_N\times \{1\}]$. Thus, the codes obtained when evaluating in this set would be monomially equivalent to the ones obtained in the affine case when evaluating in $Y_N$. As we said before, this does not mean that their subfield subcodes are monomially equivalent. Nevertheless, our experiments show that the parameters that we obtain when evaluating in the set $\mathbb{X}_N^*$ are strictly worse than the ones obtained in the affine case with $Y_N$. Hence, in what follows we are going to focus on evaluating in the set $\X_N$, although we note that the theory we are going to develop can be adapted for the set $\X^*_N$ as well.

We denote the standard representatives of $\X_N$ by $X_N$, and we also denote $\D:=\PRS(X_N,\Delta)$. With this notation, doubly extended Reed-Solomon codes are denoted by $\PRS(q^s,\Delta_d)$. Similarly to the case of doubly extended Reed-Solomon codes, given $1\leq d \leq N$, the parameters of the code $\PRS(N,\Delta_d)$ are $[N+1,d+1,N-d+1]$. In general, for the codes $\D$ we have the parameters $[N+1,\abs{\Delta},\geq N-d(\Delta)+1]$, where the bound for the minimum distance is given by the smallest doubly extended Reed-Solomon code that contains $\D$.

\section{Subfield subcodes of codes over the projective line}\label{secprimario}

Let $\F_{q^s}\supset \F_q$ and $N$ such that $N-1\mid q^s-1$. In this section we want to obtain bases for the subfield subcodes of the codes $\D$ with respect to this extension, which we will denote by $\DS:=\D\cap \fq$. Given $f\in S$, we say that $f$ evaluates to $\fq$ in $X_N$ whenever $f(Q)\in \fq$ for all $Q\in X_N$ (similarly for polynomials in $R$ evaluating in $Y_N$). The following lemma gives the key idea in order to obtain bases for $\DS$.

\begin{lem}\label{lemassc}
Let $X_N\subset P^1$. Then $f\in S$ evaluates to $\fq$ in $X_N$ $\iff$ $f(1,x_1)$ evaluates to $\fq$ in $Y_N$ and $f(0,1)$ is in $\fq$.
\end{lem}

We will see now that we can take advantage of the knowledge from the affine case in Theorem \ref{baseafin} by homogenizing and using Lemma \ref{lemassc}. Given a degree $d$ and a polynomial $f(x)\in R$ with $\deg(f)\leq d$, its homogenization up to degree $d$ is the homogeneous polynomial $f^h(x_0,x_1):=x_0^d f(x_1/x_0)\in S_d$. Unless stated otherwise, when we consider the code $\D$, we are always going to assume that we are homogenizing up to degree $d=d(\Delta)$.

For a polynomial $f\in \fq[x_1]$, we choose $\mathcal{T}_a(f)$ as the representative of the class in $\F_{q^s}[x_1]/I(Y_N)$ which has the exponents of each monomial reduced modulo $q^s-1$. Given $d\geq 1$, if the degree of $\mathcal{T}_a(f)$ is lower than $d$, then we define $\mathcal{T}_a^h(f):=\left(\mathcal{T}_a(f)\right)^h$, which we call homogenized trace. If we consider one of the traces that appear in Theorem \ref{baseafin}, its homogenized trace automatically satisfies that, when setting $x_0=1$, the resulting polynomial evaluates to $\fq$ in $Y_N$, i.e., the first condition from Lemma \ref{lemassc} is satisfied. However, the second condition, which means that the coefficient of $x_1^{d}$ in the homogenized trace must be in $\fq$, might not be satisfied. Because of this, the projective case is more involved than the affine case, as we will see in the next example.

\begin{ex}\label{extrazas}
We continue with Example \ref{excyclo}. By Theorem \ref{baseafin}, the following polynomial associated to $\II_1$ evaluates to $\F_3$:
$$
\mathcal{T}_{1}(x)=x+x^3.
$$
Let $d=3$ (the degree up to which we homogenize). If we consider the polynomial $f=\mathcal{T}^h_{1}(x_1)=x_0^2x_1+x_1^3$, this is a homogeneous polynomial of degree 3 such that $f(1,x_1)$ takes the same values as $\mathcal{T}_{1}(x_1)$ in $\F_9$, and $f(0,1)=1\in \F_3$. By Lemma \ref{lemassc}, we know that $f$ evaluates to $\F_3$ when evaluating in $P^1$. 

If $\xi$ is a primitive element in $\F_9$, by Theorem \ref{baseafin}, the following polynomial also evaluates to $\F_3$:
$$
\mathcal{T}_{1}(\xi x)=\xi x+\xi ^3x^3.
$$
However, if we consider $g=\mathcal{T}^h_{1}(\xi x_1)=\xi x_0^2x_1+\xi ^3x_1^3$, we see that $g(0,1)=\xi^3\not\in\F_3$. Therefore $g$ does not evaluate to $\F_3$.
\end{ex}

\begin{rem}
If we have $f\in S_d$ which evaluates to $\fq$, then $x_0f\in S_{d+1}$ also evaluates to $\fq$. Moreover, if $f(1,x_1)$ evaluates to $\fq$ in $Y_N$, then $g=x_0f\in S_{d+1}$ evaluates to $\fq$ in $X_N$, even if $f$ does not, because $g(1,x_1)=f(1,x_1)$, which evaluates to $\fq$, and $g(0,1)=0\in \fq$. This already gives a hint about the fact that the sequence of dimensions of the subfield subcodes is going to be non-decreasing.
\end{rem}

With Lemma \ref{lemassc}, we can consider polynomials in one variable that evaluate to $\fq$ in order to obtain polynomials in $S_d$ that evaluate to $\fq$ in $X_N$ in some cases. One could also consider the polynomials in two variables that evaluate to $\fq$ when evaluating in the points of $\mathbb{A}^2$. All of those polynomials are going to evaluate to $\fq$ when evaluating in points of $P^1$. However, there are bivariate polynomials that evaluate to $\fq$ in $P^1$, but not in $\mathbb{A}^2$. For example, in Example \ref{extrazas} we consider $f=x_0^2x_1+x_1^3$, which evaluates to $\F_3$ over $P^1$, but if we consider this polynomial over $\A^2$, then it is clear that it does not evaluate to $\F_3$. For example, if $\xi$ is a primitive element in $\F_9$, $f(0,\xi)=\xi^3\not\in \F_3$.

The following result shows how to use the previous ideas to obtain a basis for $\DS$.

\begin{thm}\label{baseproyectiva}
Let $\Delta$ be a nonempty subset of $\{0,1,\dots,N-1\}$ and let $d=d(\Delta)$. Set $\xi_b$ a primitive element of the field $\F_{q^{n_b}}$. A basis for $\PRS(N,\Delta)_q$ is given by the image by $\ev_{X_N}$ of the following polynomials.

If $\II_d\subset \Delta$:

$$
\bigcup_{b\in \mathcal{B}\mid \II_b\subset \Delta, b< d}\{\mathcal{T}_b^h(\xi_b^r x_1^b)\mid 0\leq r\leq n_b-1 \} \cup \{\mathcal{T}_d^h(x_1^d) \}.
$$

If $\II_d\not\subset\Delta$:

$$
\bigcup_{b\in \mathcal{B}\mid \II_b\subset \Delta}\{\mathcal{T}_b^h(\xi_b^r x_1^b)\mid 0\leq r\leq n_b-1 \}.
$$

\end{thm}
\begin{proof}
If we consider 
$$
\bigcup_{b\in \mathcal{B}\mid \II_b\subset \Delta, b< d}\{\mathcal{T}_b^h(\xi_b^r x_1^b)\mid 0\leq r\leq n_b-1 \},
$$
these are functions which have linearly independent evaluations, because when evaluating in $[\{1\}\times X_N]$ they are linearly independent by Theorem \ref{baseafin}. These polynomials do not have the monomial $x_1^d$ in their support. Therefore, by Lemma \ref{lemassc}, they evaluate to $\fq$ in $X_N$. 

If $\II_d\not\subset \Delta$, we are going to see that the evaluation of these polynomials generates the whole subfield subcode. Let $S_{d,\Delta}\subset S_d$ be the linear space generated by $\{x_0^{d-i}x_1^i\mid i \in \Delta\}$, and let $f\in S_{d,\Delta}$ be such that its evaluation is in $\DS$. If $f(0,1)=0$, then, using Theorem \ref{baseafin}, we know that we can generate the evaluation of $f$ with these polynomials. On the other hand, we claim that $f(0,1)\neq 0$ cannot happen in this case, which means that the image by the evaluation map of the stated polynomials generate the whole subfield subcode. If we had $f(0,1)\neq 0$, that would imply that $f(1,x_1)$ has the monomial $x_1^d$ in its support. However, if $\II_d\not\subset\Delta$, then we know that there is at least one $a_1\in \II_d$ which is not in $\Delta$. Therefore, we cannot obtain the monomial $x_1^{a_1}$ in the support of $f(1,x_1)$ because using Theorem \ref{baseafin} in $Y_N$, once you have $x_1^d$ in the support of $f(1,x_1)$, you should have $x_1^a$ in its support for all $a\in\II_d$ because $f(1,x_1)$ should be a linear combination of traces. Therefore, $f(0,1)\neq 0$ is not possible in this case, and the stated polynomials generate the whole subfield subcode.

In the case $\II_d\subset \Delta$ we have that $d\in\B$, i.e., there is a minimal cyclotomic set whose maximal representative is equal to $d$. By Lemma \ref{lemassc}, we have that $\mathcal{T}_d^h(x_1^d)$ evaluates to $\fq$, and it is linearly independent from the other polynomials that we consider because it is the only one that takes a nonzero value at $[0:1]$.

We are going to show now that the evaluation of the given set of polynomials generates the whole code in this case. Let $f\in  S_{d,\Delta} $ such that $f$ evaluates to $\fq$. By Lemma \ref{lemassc}, $f(0,1)$ is in $\fq$. Hence, we can subtract $\mathcal{T}_d^h(x_1^d)$ multiplied by $f(0,1)\in \fq$ and the evaluation would still be in $\fq$. Therefore, we can assume that $f$ does not have the monomial $x_1^d$ in its support, i.e., $f(0,1)=0$. Then we can use the affine case and argue that if $f(1,x_1)$ evaluates to $\fq$, by Theorem \ref{baseafin} it must be a linear combination of the polynomials in 
$$
\bigcup_{b\in \mathcal{B}\mid \II_b\subset \Delta, b< d}\{\mathcal{T}_b(\xi_b^r x_1^b)\mid 0\leq r\leq n_b-1 \}.
$$
The homogenized polynomials that we consider have the same evaluation as these polynomials in $[\{1\}\times Y_N]$, which completes the proof.
\end{proof}

\begin{rem}
We note that we are obtaining a basis which is the image by the evaluation map of some homogeneous polynomials of degree $d$, which we already knew that should be possible because $\DS\subset \D$. 
\end{rem}

\begin{ex}\label{exbaseproyectiva}
We continue with Examples \ref{excyclo} and \ref{extrazas}. We consider $N=9$ and $\Delta=\{0,1,2,3\}$, which means that we have $d(\Delta)=3$. Looking at the cyclotomic sets from Example \ref{excyclo}, we see that $\II_0\cup \II_1\subset \Delta$ (and these are the only complete minimal cyclotomic sets in $\Delta$). By Theorem \ref{baseproyectiva}, taking into account that in this case we have $\II_3=\II_d\subset \Delta$, we see that the evaluation of the following polynomials is a basis for $\PRS(9,\Delta)_3$:
$$
\mathcal{T}_0^h(x_1^0)=x_0^3, \mathcal{T}_3^h(x_1^3)=x_0^2x_1+x_1^3.
$$
We note that the second polynomial is precisely the polynomial $f$ in Example \ref{extrazas}.

If we take $\Delta=\{0,1,2,3,4\}$, then we have $d(\Delta)=4$ and $\II_0\cup \II_1\cup \II_4\subset \Delta$. By Theorem \ref{baseproyectiva}, the evaluation of the following polynomials is a basis for $\PRS(9,\Delta)_3$:
$$
\mathcal{T}_0^h(x_1^0)=x_0^4, \mathcal{T}_3^h(x_1^3)=x_0^3x_1+x_0x_1^3,\mathcal{T}_3^h(\xi x_1^3)=\xi^3 x_0^3x_1+\xi x_0 x_1^3, \mathcal{T}_4^h(x_1^4)=x_1^4.
$$
\end{ex}

\begin{cor}\label{dimprimario}
The dimension of $\PRS(N,\Delta)_q$ is the following:
$$\dim \PRS(N,\Delta)_q=
\begin{cases}
\displaystyle \sum_{b\in \mathcal{B}:\II_b\subset \Delta} n_b - (n_d-1)= \sum_{b\in \mathcal{B}:\II_b\subset \Delta,b<d} n_b +1 &\text{ if }  \II_d\subset \Delta\\
\displaystyle \sum_{b\in \mathcal{B}:\II_b\subset \Delta} n_b &\text{ otherwise }
\end{cases}
$$
\end{cor}
\begin{rem}
Let $d=d(\Delta)$. If $\II_d\subset\Delta $, we have dimension 1 more than in the affine case with $\Delta\setminus \{d\}$. On the other hand, if  $\II_d\not\subset\Delta$, we obtain a degenerate code with a 0 at the point $[0:1]$. Therefore, the interesting case is when $\II_d\subset\Delta $, which is the one in which we are going to mainly focus in what follows.
\end{rem}

With respect to the minimum distance, if we denote by $\wt(C)$ the minimum distance of a code $C\subset \F_{q^s}^n$, we have $\wt(\D)\geq N-d(\Delta)+1$, which implies that $\wt(\DS)\geq N-d(\Delta)+1$ because $\DS\subset \D$. For the case of subfield subcodes of doubly extended Reed-Solomon codes we obtain the following corollary. 

\begin{cor}\label{parametersPRS}
Let $d\in \mathcal{B}$. The parameters of $\PRS(q^s,\Delta_d)_q$ are $[q^s+1,\displaystyle\sum_{b\in \mathcal{B}:b<d} n_b +1,\geq q^s-d+1]$. Moreover, the first nontrivial (dimension higher than 1) subfield subcode is obtained when $d=q^{s-1}$.
\end{cor}
\begin{proof}
The parameters are a special case of the previous results and discussions. For the last statement, it is clear that $q^s/q=q^{s-1}$ is the lowest possible element in $\mathcal{B}$ (besides 0), and $d=q^{s-1}$ is the first degree such that $\II_1=\{1,q,q^2,\dots,q^{s-1}\}\subset \Delta_d$.
\end{proof}

The bound used for the minimum distance of the subfield subcodes of doubly extended Reed-Solomon codes is sharp in all cases we have checked with $d\in \B$. The codes obtained in this way have one more length and dimension than in the affine case, with the same minimum distance.

\begin{ex}
If we look at the results from Example \ref{exbaseproyectiva}, we see that we obtained dimension $2$ and $4$ for $\PRS(9,\Delta_3)_3$ and $\PRS(9,\Delta_4)_3$. These are the values obtained with Corollary \ref{parametersPRS}, because $2=n_0+1$ and $4=n_0+n_3+1$. We would obtain codes with parameters $[10,2,7]$ and $[10,4,6]$ over $\F_3$.
\end{ex}

\section{Dual codes of the previous subfield subcodes}\label{secdual}

In order to compute the dual codes of the previous subfield subcodes, we are going to use Delsarte's Theorem \cite{delsarte}. 

\begin{thm}\label{delsarte}
Let $C\subset \F_{q^s}^n$ be a linear code.
$$
(C\cap \fq^n)^\perp=\Tr(C^\perp),
$$
where $\Tr:\F_{q^s} \rightarrow \fq$, which maps $x$ to $x+x^q+\cdots + x^{q^{s-1}}$, is applied componentwise to $C^\perp$.
\end{thm}

In order to use this result, we need to compute the dual of the codes $\PRS(N,\Delta)$. It is well known that $\PRS(q^s,\Delta_d)^\perp=\PRS(q^s,\Delta_{q^s-1-d})$ (the dual of a doubly extended Reed-Solomon code is another doubly extended Reed-Solomon code). However, computing the dual of the codes $\PRS(N,\Delta)$ in general can be involved. Nevertheless, we can easily compute the dual in some cases. In order to do so, we are going to state the metric structure of these codes first. Part of the following result already appears in  \cite[Prop. 1]{galindostabilizer} and \cite[Lem. 7.1]{hiramdualevaluation}.

\begin{lem}\label{sumasubgrupo}
Let $\gamma$ be a non-negative integer, and $N$ such that $N-1\mid q^s-1$. We consider the monomial $x^\gamma\in\F_{q^s}[x]$. We have the following:
$$
\sum_{z\in Y_N}x^\gamma(z)=
\begin{cases}
N &\text{ if } \gamma=0,\\
0 &\text{ if } \gamma>0 \text{ and } \gamma\not\equiv 0 \bmod (N-1),\\
N-1 &\text{ if } \gamma>0 \text{ and } \gamma\equiv 0 \bmod (N-1).
\end{cases}
$$
\end{lem}
\begin{proof}
Let $\xi\in \F_{q^s}$ be an element of order $N-1$, which exists because $N-1\mid q^s-1$. Then $Y_N=\{\xi^0,\xi^1,\dots,\xi^{N-2}\}\cup\{0\}$. If $\gamma=0$, $x^\gamma=1$, and the sum is equal to $\abs{Y_N}=N$. If $\gamma>0$ and $\gamma\equiv 0 \bmod (N-1)$, then $x^\gamma(z)=1$ for all $z\in Y_N\setminus \{0\}$, and $\sum_{z\in Y_N}x^\gamma(z)=\abs{Y_N}-1=N-1$. Finally, if $\gamma>0$ and $\gamma\not \equiv 0 \bmod (N-1)$, we have
$$
\sum_{z\in Y_N}x^\gamma(z)=\sum_{i=0}^{N-2}(\xi^{i})^\gamma=\frac{\xi^{\gamma(N-1)}-1}{\xi^\gamma -1}=0.
$$
\end{proof}

\begin{prop}\label{metric}
Let $x_0^{\alpha_0}x_1^{\alpha_1}$ and $x_0^{\beta_0}x_1^{\beta_1}$ be two monomials in $\fqs[x_0,x_1]$ of degree $d_\alpha$ and $d_\beta$, respectively. Then  we have the following for the product of the evaluations over $X_N$.
If $\alpha_1+\beta_1=0$:
$$
\ev_{X_N}(x_0^{\alpha_0}x_1^{\alpha_1})\cdot \ev_{X_N}(x_0^{\beta_0}x_1^{\beta_1})=
\begin{cases}
N+1 &\text{ if }  \alpha_0+\beta_0=0, \\
N &\text{ if } \alpha_0+\beta_0>0 .
\end{cases}
$$
If $\alpha_1+\beta_1>0$:
$$
\ev_{X_N}(x_0^{\alpha_0}x_1^{\alpha_1})\cdot \ev_{X_N}(x_0^{\beta_0}x_1^{\beta_1})=
\begin{cases}
N &\text{ if } \alpha_1+\beta_1\equiv 0\mod (N-1), \alpha_0+\beta_0=0, \\
N-1 &\text{ if } \alpha_1+\beta_1\equiv 0\mod (N-1), \alpha_0+\beta_0>0, \\
1 &\text{ if } \alpha_1+\beta_1\not\equiv 0\mod (N-1), \alpha_0+\beta_0=0 ,\\
0 &\text{ if } \alpha_1+\beta_1\not\equiv 0\mod (N-1), \alpha_0+\beta_0>0. 
\end{cases}
$$
\end{prop}
\begin{proof}
First, we expand the scalar product as a sum over $X_N=\{[1:z]\mid z\in Y_N\}\cup \{[0:1]\}\subset P^1$:
$$
\ev_{X_N}(x_0^{\alpha_0}x_1^{\alpha_1})\cdot \ev_{X_N}(x_0^{\beta_0}x_1^{\beta_1})=\sum_{P\in X_N}x_0^{\alpha_0+\beta_0}x_1^{\alpha_1+\beta_1}(P)=\sum_{z\in Y_N}z^{\alpha_1+\beta_1}+\epsilon,
$$
where $\epsilon$ is equal to $1$ if $\alpha_0+\beta_0=0$ and equal to $0$ if $\alpha_0+\beta_0>0$ (corresponding to the evaluation at $[0:1]$). The result is obtained by using Lemma \ref{sumasubgrupo}.
\end{proof}

If $p$ does not divide $N$, we have that the evaluation of $x_1^{\alpha_1}$ with $\alpha_1>0$ is not orthogonal to the evaluation of $x_1^{\beta_1}$ for any $\beta_1$. This means that the dual code $\Dtt$ does not have a basis obtained by the evaluation of monomials unless $\wt(\D)=1$. This is because if we have $\wt(\D)>1$, then $\Dtt$ cannot be degenerate. In particular, there must be a vector in $\Dtt$ such that the coordinate associated to the point $[0:1]$ is nonzero, which is obtained by evaluating a polynomial with some power of $x_1$ in its support, but it cannot be just a single power of $x_1$ because its evaluation would not be orthogonal to the evaluation of $x_1^{d}$. 
Hence, the dual code is not generated by the image by the evaluation map of monomials.

When $p\mid N$, as the next result shows, the previous result gets simplified, and in Proposition \ref{propdual} we will see that in this case the dual code can be generated by the evaluation of monomials.

\begin{cor}\label{metricfacil}
If $p\mid N$, then:
$$
\ev_{X_N}(x_0^{\alpha_0}x_1^{\alpha_1})\cdot \ev_{X_N}(x_0^{\beta_0}x_1^{\beta_1})=
\begin{cases}
1 &\text{ if } \alpha_1+\beta_1=0, \alpha_0+\beta_0=0 \text{ or } \\
&\hspace{0.5cm}\alpha_1+\beta_1\not\equiv 0\mod (N-1), \alpha_0+\beta_0=0, \\
-1 &\text{ if } \alpha_1+\beta_1\equiv 0\mod (N-1), \alpha_i+\beta_i>0,i=0,1,\\
0 &\text{ otherwise.} 
\end{cases}
$$
\end{cor}

\begin{rem}
One way to have $p\mid N$ is to consider a subfield of $\fqs$, in which case we are going to obtain a doubly extended Reed-Solomon code over that subfield. However, we may also have $p\mid N$ for different subgroups of $\fqs^*$. For example, if we consider $q^s=2^4=16$, then $5$ divides $q^s-1$. Therefore, we can take $N=6$, which is divisible by $2$, but $Y_6$ is not a subfield of $\F_{16}$.
\end{rem}

For obtaining a basis for the dual code we will need to work with non-homogeneous polynomials. In order to understand linear independence in that situation we are going to introduce now a universal Gröbner basis for the vanishing ideal $I(X_N)$. Particular cases of the following result were already present in \cite{decodingRMP}.

\begin{prop}\label{grobnerP1}
A universal Gröbner basis for the ideal $I(X_N)$ is:
$$
I(X_N)=\langle x_0^2-x_0,x_1^N-x_1,(x_0-1)(x_1-1) \rangle.
$$
Therefore, $\ini(I(X_N))=\langle x_0^2,x_1^N,x_0x_1\rangle$ and $\{1,x_0,x_1,x_1^2,\dots,x_1^{N-1} \}$ is a basis for the quotient ring $S/I(X_N)$.
\end{prop}
\begin{proof}
First, we are going to show that these polynomials generate the vanishing ideal $I(X_N)$. Given any point in $X_N$, it is clear that it satisfies the equations. Reciprocally, any point satisfying this equations, because of the generator $x_0^2-x_0$, must have the first coordinate equal to 0 or 1. If the first coordinate is equal to 0, because of the generator $(x_0-1)(x_1-1)$, the last coordinate must be 1, i.e., it must be the point $[0:1]\in X_N$. If the first coordinate is equal to 1, then, because of the generator $x_1^N-x_1$, the second coordinate is in $Y_N$, which means that the point is in $X_N$ as well.

We have proved that the variety defined by this ideal is $X_N$. It is clear that the variety defined by this ideal over the algebraic closure $\overline{\F_{q^s}}$ is the same as the variety defined over $\F_{q^s}$. By Seidenberg's Lemma \cite[Prop. 3.7.15]{kreuzer1}, this ideal is radical. Therefore, by Hilbert's Nullstellensatz applied in the algebraic closure, we have that this ideal is the vanishing ideal of the variety that it defines, i.e., is the vanishing ideal of $X_N$.

In order to show that all the S-polynomials of the generators reduce to 0, we just need to use that if the greatest common divisor of the initial monomials of two polynomials is $1$, then their $S$-polynomial reduces to 0 by \cite[Prop. 4, Chapter 2, Section 9]{cox}. In particular, if two polynomials depend on different variables, their $S$-polynomial reduces to 0. And if $f$ and $g$ share a common factor $w$, then $S(f,g)=wS(f/w,g/w)$. Using this, it is easy to see that all the S-polynomials reduce to 0 in this case, for any monomial order. Thus, these generators form a universal Gröbner basis. The initial ideal follows from this fact, and by Macaulay's classical result \cite[Thm. 15.3]{eisenbud} we obtain that the monomials not contained in the initial ideal form a basis for the quotient ring.
\end{proof}

\begin{rem}
Because of the first generator of the previous ideal, any power of $x_0$ is equivalent to $x_0$ in the quotient ring. Therefore, we have $x_0^{\alpha_0}x_1^{\alpha_1}\equiv x_0x_1^{\alpha_1}\bmod I(X_N)$ if $\alpha_0>0$. This is why we are going to assume $\alpha_0=1$ for any monomial divisible by $x_0$ in what follows, except when we want to remark that we can obtain a code by evaluating homogeneous polynomials of a certain degree.
\end{rem}

The following result allows us to express any polynomial in $S/I(X_N)$ in terms of the basis in Proposition \ref{grobnerP1}.

\begin{lem}\label{divisionP1}
Let $a_0,a_1$ be integers, with $a_0>0$. We have that
$$
x_0^{a_0}x_1^{a_1}\equiv x_0+x_1^{a_1}-1 \mod I(X_N).
$$
\end{lem}
\begin{proof}
It is easy to check that both polynomials have the same evaluation in $X_N$, which implies that they are in the same class modulo $I(X_N)$.
\end{proof}

\begin{cor}\label{baseP1}
The following monomials constitute a basis for the quotient $S/I(X_N)$:
$$
\{x_1^N,x_0,x_0x_1,\dots,x_0x_1^{N-1} \}.
$$
Moreover, every set of the form $\{ x_1^d,x_0,x_0x_1,\dots,x_0x_1^{d-1}\} $ with $1\leq d \leq N$ is linearly independent.
\end{cor}
\begin{proof}
It is easy to check that these monomials are linearly independent by Lemma \ref{divisionP1} and Proposition \ref{grobnerP1}. The fact that for $d=N$ this set is a basis follows from the cardinality of the set and the dimension of the quotient ring.
\end{proof}

Now we have the tools necessary to deal with the dual as an evaluation code over the projective line. In what follows we are going to assume that $p\mid N$. This is because, by Corollary \ref{metricfacil}, the metric structure is going to be similar to the one of doubly extended Reed-Solomon codes, and in this case the dual code will be generated by the evaluation of monomials. For the following result it will be useful to introduce the definition $\Delta^\perp=\{ \alpha\in \{0,1,\dots,N-1\}\mid \alpha\neq N-1-h,h\in \Delta \}$. 

\begin{prop}\label{propdual}
Let $N$ be a non-negative integer such that $N-1\mid q^s-1$ and $p\mid N$. Let $\Delta\subset \{0,1,\dots,N-1\}$ and let $d=d(\Delta)$. Then $\PRS(N,\Delta)^\perp$ has a basis obtained by taking the image by $\ev_{X_N}$ of the following monomials:
\begin{equation}\label{basedual}
\{x_0x_1^\alpha\mid \alpha\in \Delta^\perp \}\cup \{x_1^{N-1-d} \}.
\end{equation}
Moreover, if $N-1\not \in \Delta$, we can also obtain the same basis by taking the image by $\ev_{X_N}$ of the following monomials of degree $2(N-1)-d$ (which allows us to get the dual code as an evaluation code of homogeneous polynomials):
\begin{equation}\label{basedual2}
\{x_0^{2(N-1)-d-\alpha}x_1^\alpha\mid \alpha\in \Delta^\perp \}\cup \{x_1^{2(N-1)-d} \}.
\end{equation}
If $N-1\in\Delta$, then the following set of homogeneous polynomials of degree $2N-1$ give the same image as the set in item \emph{(\ref{basedual})}:
\begin{equation}\label{basedual3}
\{x_0^{2N-1-\alpha}x_1^\alpha\mid \alpha\in \Delta^\perp \}\cup \{x_1^{2N-1}+x_0^{2N-1}-x_0^{N-1}x_1^N \}.
\end{equation}
\end{prop}

\begin{proof}
Using Corollary \ref{metricfacil} it is easy to see that the evaluation of the monomials in (\ref{basedual}) is orthogonal to the vectors in $\D$. When $N-1\not \in \Delta$, using Lemma \ref{divisionP1} it is easy to see that the evaluation of these monomials is linearly independent, and the dimension of this subspace is the same as the dimension of the dual code. If $N-1\in \Delta$, then $x_1^{N-1-d}=1$, and it is easy to see that the monomials that we obtain are linearly independent and generate the dual code. When $N-1\not\in \Delta$, the evaluation of the set (\ref{basedual2}) is clearly the same. Finally, if $N-1\in\Delta$, we have that
$$
x_1^{2N-1}+x_0^{2N-1}-x_0^{N-1}x_1^N\equiv x_1+x_0-x_0x_1\equiv 1\mod I(P^1).
$$
Therefore, the evaluation of the set (\ref{basedual3}) is the same as the one obtained with (\ref{basedual}).
\end{proof}

We have the next result for the case when $p\mid N$, which generalizes what we know about the duality in the case of doubly extended Reed-Solomon codes. We note that, as we are evaluating all the monomials of degree $d$ in the next result, and the set of evaluation points is a complete intersection, the theory from \cite{duursmacompleteintersection} and \cite{gonzalezdualsomereedmuller} could also be used to study the codes $\PRS(N,\Delta_d)$ and their duals.

\begin{cor}\label{dualfacil}
Let $\Delta_d=\{0,1,\dots,d\}$ and $\Delta_{N-1-d}=\{0,1,\dots,N-1-d\}$. If $p\mid N$, then we have that $\PRS(N,\Delta_d)^\perp=\PRS(N,\Delta_{N-1-d})$.
\end{cor}
\begin{proof}
We can consider the monomials in (\ref{basedual}), homogenizing up to degree $N-1-d$ with the variable $x_0$. Taking into account that in this case $\Delta^\perp\cup \{N-1-d\}=\Delta_{N-1-d}$ we obtain the result.
\end{proof}

With the evaluation map $\ev_{X_N}$, if we consider the trace function $T:S\rightarrow S$, defined by $f\rightarrow f+f^q+\cdots+f^{q^{s-1}}$, then it is easy to verify that $\ev_{X_N} \circ\: T=\Tr \circ \ev_{X_N}$ ($\Tr$ was defined in Theorem \ref{delsarte}). Then we see that if $\PRS(N,\Delta)^\perp=\ev_{X_N} (\langle \{  f_1,f_2,\dots,f_l  \} \rangle)$, using Theorem \ref{delsarte} and the previous observation we get that 
$$
\DT:=(\PRS(N,\Delta)_q)^\perp=\Tr(\ev_{X_N} (\langle \{  f_1,f_2,\dots,f_l  \} \rangle))=\ev_{X_N}(T( \langle \{  f_1,f_2,\dots,f_l  \} \rangle)).
$$

\begin{rem}\label{basetrazas}
Taking into account that $T$ is linear, then it is clear that in this situation $\DT$ is spanned by the image by the evaluation of the polynomials $T(\gamma f_i)$, $\gamma \in \fqs$, $i=1,\dots,l$. 
\end{rem}

Even if $f_i$, for $i=1,\dots,l$, are monomials, the dual code will be generated by traces of those monomials by Remark \ref{basetrazas}, which in general are going to be non-homogeneous polynomials. We have introduced the vanishing ideal from Proposition \ref{grobnerP1} precisely to understand linear independence of sets of monomials of different degree over $X_N$. In order to state a basis for $\DT$ we will need the following lemma.

\begin{lem}\label{lemaciclotomicos}
Let $\Delta\subset \{0,1,\dots,N-1\}$ and $0<a\neq N-1$. Then we have that $\II_a\subset \Delta \iff \abs{\II_{N-1-a}\cap \Delta^\perp}=0$.
\end{lem}
\begin{proof}
It is clear that we have a bijection between $\II_a$ and $\II_{-a}$, given by $h\mapsto -h$. In $\Z_N$ we have that $-h\equiv N-1-h \mod N-1$ if $h\neq 0$. Hence, we get a bijection between $\II_a$ and $\II_{N-1-a}$ given by $h\mapsto N-1-h$. Because of the definition of $\Delta^\perp$, we see that if $h\in \Delta$, then $N-1-h\not \in \Delta^\perp$. Thus, it is clear that if $\II_a\subset \Delta$, then $\abs{\II_{N-1-a}\cap \Delta^\perp}=0$, and vice versa.
\end{proof}

\begin{thm}\label{baseproyectivadual}
Let $\Delta$ be a nonempty subset of $\{0,1,\dots,N-1\}$ and let $d=d(\Delta)$. Set $\xi_a$ a primitive element of the field $\F_{q^{n_a}}$ with $\mathcal{T}_a(\xi_a)\neq 0$ (this can always be done \cite{cohen}). A basis for $\DT$ is given by the image by $\ev_{X_N}$ of the following polynomials.

If $\II_d\subset \Delta$:
$$
\bigcup_{a\in \mathcal{A}\mid \II_a\cap \Delta^\perp \neq \emptyset}\{\mathcal{T}_a(\xi_a^r x_0 x_1^a)\mid 0\leq r\leq n_a-1 \}\cup \{ \mathcal{T}_{N-1-d}(\xi_{N-1-d}^r x_1^{N-1-d}) \mid 0\leq r \leq n_{N-1-d}-1\}.
$$

If $\II_d\not\subset \Delta$:

$$
\bigcup_{a\in \mathcal{A}\mid \II_a\cap \Delta^\perp \neq \emptyset}\{\mathcal{T}_a(\xi_a^r x_0 x_1^a)\mid 0\leq r\leq n_a-1 \}\cup \{ \mathcal{T}_{N-1-d}(\xi_{N-1-d} x_1^{N-1-d})\}.
$$
\end{thm}
\begin{proof}
In Remark \ref{basetrazas} we saw that it is enough to consider the traces of multiples of the monomials whose images span $\PRS(N,\Delta)^\perp$. Therefore, we have that the traces of multiples of the monomials in (\ref{basedual}) span $\Tr(\PRS(N,\Delta)^\perp)=\DT$. Moreover, it is enough to consider the following traces for $\II_a$ with $\II_a\cap \Delta^\perp\neq \emptyset$
$$
\{\mathcal{T}_a(\xi_a^r x_0 x_1^a), 0\leq r\leq n_a-1\}
$$
because they are linearly independent (a dependence relation would give a polynomial relation on $\xi_a$ of degree less than $n_a$) and there are $n_a$ of them, which is the maximum dimension that we can get with $n_a$ monomials. The same reasoning shows that it is enough to consider the following traces for the monomial $x_1^{N-1-d}$:
\begin{equation}\label{trazasx1}
    \{\mathcal{T}_{N-1-d}(\xi_{N-1-d}^r x_1^{N-1-d}), 0\leq r\leq n_{N-1-d}-1\},
\end{equation}
which are linearly independent between them as well. 

If $\II_d\subset \Delta$, by Lemma \ref{lemaciclotomicos} we have that $\abs{\II_{N-1-d}\cap \Delta^\perp}=\emptyset$. Hence, when we consider all of these sets of polynomials together, they are independent because between sets corresponding to different cyclotomic sets $\II_a$ we have polynomials with disjoint support (the monomials that we are considering are linearly independent in $S/I(X_N)$ by Corollary \ref{baseP1}).

On the other hand, when $\II_d\not \subset \Delta$, by Lemma \ref{lemaciclotomicos} we know that there is at least one element $h\in \II_{N-1-d}\cap \Delta^\perp$. The argument for the previous case works in this case, except when considering the traces of polynomials associated to $\II_{N-1-d}$ and the polynomials in (\ref{trazasx1}), because by Lemma \ref{divisionP1} we will have the same powers of $x_1$. However, if from the later set of polynomials we only consider $\mathcal{T}_{N-1-d}(\xi_{N-1-d}x_1^{N-1-d})$, then the linear independence is clear because this polynomial is equal to $\mathcal{T}_{N-1-d}(\xi_{N-1-d})\neq 0$ at $[0:1]$ (because of the choice of the primitive elements), while the rest of polynomials that we are considering are $0$ at $[0:1]$. Moreover, with these polynomials we can generate the rest of the polynomials in (\ref{trazasx1}) taking into account Lemma \ref{divisionP1}:

\begin{align*}
    \mathcal{T}_a(\xi_a^r x_0 x_1^a)=& \xi_a^r(x_0+x_1^a-1)+\xi_a^{qr}(x_0+x_1^{qa}-1)+\cdots+\xi_a^{q^{n_a-1}r}(x_0+x_1^{q^{n_a-1}a}-1)=\\ &\mathcal{T}_a(\xi_a^r)(x_0-1)+\mathcal{T}_a(\xi_a^r x_1^a).
\end{align*}

With $r=1$ we see that we can generate $(x_0-1)$ with the polynomials we are considering, and with $(x_0-1)$ we can generate the rest of polynomials in (\ref{trazasx1}) because $\mathcal{T}_a(\xi_a^r)\in \F_q$.
\end{proof} 

In the case $\II_{d(\Delta)}\not\subset \Delta$ of the previous result, we have seen that we can generate $(x_0-1)$. The evaluation of this polynomial on $P^1$ gives a codeword with Hamming weight $1$, which means that $\DT$ has minimum distance 1. This is equivalent to having that $\PRS(N,\Delta)_q$ is a degenerate code (it has a common zero in the coordinate associated to the point $[0:1]$). Once again, we see that the interesting case for us is when $\II_{d(\Delta)}\subset \Delta$.

\begin{ex}
We continue with example \ref{extrazas}. Let $\Delta_4=\{0,1,2,3,4\}$, which implies $d(\Delta_4)=4$. We are going to obtain a set of polynomials such that its image by the evaluation map is a basis for $(\PRS(9,\Delta_4)_3)^\perp$. We have that $\Delta_4^\perp=\{0,1,2,3\}$. The minimal cyclotomic sets $\II_a$ with $\II_a\cap\Delta^\perp\neq \emptyset$ are $\II_0,\II_1$ and $\II_2$. As in the previous examples, if $\xi$ is a primitive element of $\F_9$, by Theorem \ref{baseproyectivadual}, we obtain the following set of polynomials:
$$
\begin{aligned}
\T_0(x_0)=x_0, \T_1(x_0x_1)=x_0x_1+x_0^3x_1^3, \T_1(\xi x_0x_1)=\xi x_0x_1+\xi x_0^3x_1^3\\
\T_2(x_0x_1^2)=x_0x_1^2+x_0^3x_1^6, \T_2(\xi x_0x_1^2)=\xi x_0x_1^2+\xi^3 x_0^3x_1^6, \T_{4}(x_1^4)=x_1^4.
\end{aligned}
$$
In all the previous expressions, we can reduce the exponent of $x_0$ to 1 and the evaluation would not change.
\end{ex}

As a consequence of Theorem \ref{baseproyectivadual}, we obtain directly an explicit formula for the dimension of $\DT$ without using the dimension of the primary codes from Corollary \ref{dimprimario}.

\begin{cor}\label{dimdual}
Let $\Delta\subset \{0,1,\dots,N-1\}$ and let $d=d(\Delta)$. The dimension of $\DT$ is equal to:
$$
\dim \DT=
\begin{cases}
\displaystyle \sum_{a\in \mathcal{A}\mid \II_a\cap \Delta^\perp\neq \emptyset}n_a +n_d &\text{ if } \II_d\subset \Delta\\
\displaystyle \sum_{a\in \mathcal{A}\mid \II_a\cap \Delta^\perp\neq \emptyset}n_a +1 &\text{ if } \II_d\not\subset\Delta
\end{cases}
$$
\end{cor}

Now we are going to turn our attention to the minimum distance of the dual code $\DT$. In the affine case, a BCH-type bound has been used frequently for the minimum distance of the duals of the subfield subcodes of Reed-Solomon codes. If one considers the code $\E$ with $\Delta=\II_{a_0}\cup \II_{a_1}\cup \cdots\cup \II_{a_l}$ a union of cyclotomic sets, then this code is Galois invariant in the sense of \cite{bierbrauercyclic}, i.e., $\E=\left(\E\right)^q$. By \cite[Thm. 4]{bierbrauercyclic}, we have that $\Tr(\E)=\EE$. We can write Theorem \ref{delsarte} in the following way: $C^\perp\cap \fq^n=\Tr(C)^\perp$. Therefore, we have that $\left(\EE\right)^\perp=\left( \RS(N,\Delta)^{\perp}\right)_q $. For $(\RS(N,\Delta)^{\perp})_q$, it is easy to see that we have a BCH-type bound because we can consider the generator matrix of $\E$ as a pseudo-parity check matrix for the code $(\RS(N,\Delta)^{\perp})_q$ (as we did with the matrix in (\ref{pseudoparity}) for BCH codes). If we have $t$ consecutive exponents in $\Delta$, we have a Vandermonde matrix as a submatrix of the generator matrix for $\E$ and we get that $\wt\left(\ET\right)\geq t+1$.

In the projective case, arguing in a similar way, we get that, if we have $t$ consecutive exponents in $\Delta$, we have the BCH-type bound $\wt\left((\PRS(N,\Delta)^\perp)_q\right)\geq t+1$. However, even if $\Delta$ is a union of cyclotomic sets, we will see in Remark \ref{remnoGI} and Example \ref{excont} that in the projective case we do not have in general that $\D$ is Galois invariant, and thus we do not have the equality between $(\PRS(N,\Delta)^\perp)_q$ and $\DT$ in general. Nevertheless, we can still use the affine case in order to get a bound for the minimum distance. If we have a code $C\subset \fq^n$, we are going to denote by $(C,0):=\{(u_1,\dots,u_n,0)\in \fq^{n+1}\mid u=(u_1,\dots,u_n)\in C \}$. In what follows, we are  going to assume that the coordinate associated to the point $[0:1]$ is the last one. We recall that $\mathcal{A}$ (resp. $\mathcal{B}$) is the set of minimal representatives (resp. maximal representatives) of the minimal cyclotomic sets. We are going to denote $\Delta':=\Delta\setminus \{ d\}$, and $(\Delta')_\II=\bigcup_{b\in\mathcal{B},b<d\mid \II_b\subset \Delta}\II_b\subset \Delta'$ as before.

\begin{prop}\label{bchmala}
Let $\Delta\subset \{0,1,\dots,N-1\}$. We assume that $d(\Delta)\in \mathcal{B}$ with $\II_{d(\Delta)}\subset \Delta$. If $t$ is the number of consecutive exponents in $(\Delta')_\II$, then we have that $\wt\left(\DT\right)\geq t+1$.
\end{prop}

\begin{proof}
We assume that the point $[0:1]$ corresponds to the last coordinate. We have $\DS \supset (\RS(N,\Delta')_q,0)$, which implies
$$
\DT\subset (\RS(N,\Delta')_q,0)^\perp=((\RS(N,\Delta')_q)^\perp,0)+\langle (0,\dots,0,1)\rangle .
$$

We know that $(0,\dots,0,1)\not\in \DT$ because that would imply that $\DS$ is degenerate, and that is not the case because of the assumptions that we have made. Thus, any vector in $\DT$ must belong to $(\RS(N,\Delta')_q)^\perp$ after puncturing the last coordinate, and therefore the weight of any vector in $\DT$ must be at least $t+1$ because of the BCH-type bound for $(\RS(N,\Delta')_q)^\perp$.
\end{proof}

As a corollary, we have the following result about the duals of the subfield subcodes of doubly extended Reed-Solomon codes.

\begin{cor}
Let $\Delta_d=\{0,1,\dots,d\}$ with $d\in \mathcal{B}$. If $t$ is the number of consecutive exponents in $(\Delta'_d)_\II$, the parameters of $(\PRS(q^s,\Delta_d)_q)^\perp$ are $[q^s+1,\sum_{a\in \mathcal{A}\mid \II_a\cap \Delta ^\perp\neq \emptyset}n_a,\geq t+1]$.
\end{cor}

This estimate would give codes with length 1 more than in the affine case, but same dimension and same bound for the minimum distance. However, the bound for the minimum distance is not sharp in general and we are able to improve upon the affine case in many examples. For instance, in the next result we show that when $\abs{\II_d}=1$ we have a better estimate for the minimum distance.

\begin{prop}\label{galoisinvariant}
Let $\Delta\subset \{0,1,\dots,N-1\}$ such that $\abs{\II_{d(\Delta)}}=1$. Then $\PRS(N,\Delta_\II)$ is Galois invariant, we have that $(\PRS(N,\Delta_\II)^\perp)_q=(\PRS(N,\Delta_\II)_q)^\perp=\DT$, and, if there are $t$ consecutive exponents in $\Delta_\II$, the parameters of $\DT$ are $[N+1,\sum_{a\in \mathcal{A}\mid \II_a\cap \Delta_\II^\perp\neq \emptyset}n_a+1, \geq t+1]$.
\end{prop}

\begin{proof}
Let $d=d(\Delta)$. We have that $\PRS(N,\Delta_\II)$ is generated by the evaluation of monomials. Because of the fact that $\Delta_\II$ is a union of cyclotomic sets, we can divide the monomials into sets corresponding to different minimal cyclotomic sets. For $a\neq d$ we have the monomials
$$
\{x_0x_1^\alpha\mid \alpha\in\II_a\subset \Delta \}.
$$
If we consider these monomials to the power of $q$, the set remains invariant in $S/I(P^1)$ because the exponents of $x_1$ are in a cyclotomic set, and the exponent of $x_0$ does not change the evaluation. For $\II_d$ we have that $x_1^{q(N-1-d)}\equiv x_1^{N-1-d}\mod I(P^1)$ because $\abs{\II_d}=1$. Therefore, the set of monomials is invariant under taking powers of $q$, which implies that $\PRS(N,\Delta_\II)=(\PRS(N,\Delta_\II))^q$. Because of the previous discussion, we have that being Galois invariant implies in this case that $(\PRS(N,\Delta_\II)^\perp)_q=(\PRS(N,\Delta_\II)_q)^\perp$. Taking into account that $\PRS(N,\Delta_\II)_q=\DS$ because of Theorem \ref{baseproyectiva}, the parameters are clear from Theorem \ref{baseproyectivadual} and the BCH-type bound.
\end{proof}

In many situations, the previous result gives codes with higher length and dimension than in the affine case. Assuming the hypotheses of the previous result, the affine code with $\ET$ would have parameters $[N,\sum_{a\in \mathcal{A}\mid \II_a\cap \Delta_\II^\perp\neq \emptyset}n_a,\geq t+1]$, meanwhile the projective code $\DT$ would have parameters $[N+1,\sum_{a\in \mathcal{A}\mid \II_a\cap \Delta_\II^\perp\neq \emptyset}n_a+1,\geq t+1]$. 

These codes can also be compared to $(\RS(N,\Delta')_q)^\perp$, with $\Delta'=\Delta\setminus \{d(\Delta)\}$. Taking into account that $\abs{\II_d}=1$, this code has parameters $[N,\sum_{a\in \mathcal{A}\mid \II_a\cap \Delta_\II^\perp\neq \emptyset}n_a+1,\geq t'+1]$, where $t'$ is the number of consecutive exponents in $\Delta'$. We see that this code has the same dimension as $\DT$. However, the bound for the minimum distance is worse than the one for $\DT$. 

The following result shows many situations in which we can use Proposition \ref{galoisinvariant} besides the obvious case with $\Delta=\{0\}$.

\begin{lem}\label{lemaciclotomicotamano1}
Let $q>2$. If $d_{\lambda}:=\lambda(N-1)/(q-1)\in \mathbb{N}$, for some $\lambda$, $1\leq \lambda\leq q-1$, then $\abs{I_{d_\lambda}}=1$.
\end{lem}
\begin{proof}
We only have to observe that
$$
\lambda \frac{N-1}{q-1}q-\lambda \frac{N-1}{q-1}=\lambda(N-1)\equiv 0 \mod N-1.
$$
\end{proof}

\begin{rem}
If $q-1\mid N-1$, then with the previous result we obtain $q-1$ cyclotomic sets with cardinality one besides $\II_0$. For example, if $N=q^s$, then we directly have $q-1\mid N-1$. However, that is not the only case. For example we can consider $q^s=3^8$ and $N=83$. In that situation it can be checked that $q-1=2\mid 82=N-1$, and we have that $\abs{I_{41}}=1$. In this situation, when we have $q-1\mid N-1$, the previous result is actually a characterization of when we have $\abs{I_d}=1$:
$$
\begin{aligned}
\abs{I_d}=1 &\iff d q\equiv d\mod N-1\iff d(q-1)=\lambda(N-1)=\lambda(q-1)\frac{N-1}{q-1} \\
&\iff d=\lambda\frac{N-1}{q-1}, \text{ for some } 1\leq \lambda <q-1.
\end{aligned}
$$
\end{rem}

\begin{ex}\label{exdualciclo}
We consider the field extension $\F_{16}\supset \F_4$, which gives the following minimal cyclotomic sets:
$$
\begin{aligned}
\II_0=\{0\},\II_1=\{1,4\},&\II_2=\{2,8\},\II_3=\{3,12\},\II_5=\{5\},\\
\II_6=\{6,9\},\II_7=\{7,13\},&\II_{10}=\{10\},\II_{11}=\{11,14\},\II_{15}=\{15\}.
\end{aligned}
$$

We see that we have $\abs{\II_{10}}=1$. If we take $\Delta=\{0,1,4,10\}=\II_0\cup \II_1\cup\II_{10}$, then $\Delta^\perp=\{0,1,\dots,N-1\}\setminus \{ \II_{15}\cup \II_{11}\cup \II_5 \}$ and we can use Corollary \ref{dimdual} to compute the dimension. All the cyclotomic sets, besides $\II_5$, $\II_{11}$ and $\II_{15}$, have nonzero intersection with $\Delta^\perp$, and we have $\II_{d(\Delta)}=\II_{10}\subset \Delta$. Hence, by Corollary \ref{dimdual}, $\dim \DT=(n_0+n_1+n_2+n_3+n_6+n_7+n_{10})+n_{10}=13$. For the minimum distance, we have $t=2$ consecutive elements in $\Delta_\II=\Delta$, which gives the following parameters for $\DT$: $[17,13,\geq 3]$.

We can do the same for $\Delta=\{0,1,2,4,8,10\}=\II_0\cup \II_1\cup \II_2\cup \II_{10} $, and we obtain the parameters $[17,11,\geq 4]$. The true parameters are $[17,13,3]$ and $[17,11,4]$, which lengthen the parameters of the affine case $[16,12,3]$ and $[16,10,4]$. We see that the bound for the minimum distance coincides with the real minimum distance in this case. 
\end{ex}

\begin{rem}\label{remnoGI}
If we do not assume in Proposition \ref{galoisinvariant} that $\abs{\II_d}=1$, then, if $d=d(\Delta)\in\mathcal{B}$ and $\II_d\subset \Delta$ (which is the interesting case in the projective setting), we will have the evaluation of the monomial $x_1^d$ in $\D$, and also the evaluation of at least one monomial $x_0^{d-a}x_1^a$ with $a\in \II_d\setminus \{d \}$. We know that $d\equiv q^r a\mod N-1$ for some $r>0$. Thus, we have the image of $x_0^{d-q^{r-1}a}x_1^{q^{r-1}a}$ in $\D$, but if we take this monomial to the power of $q$, we get $x_0^{q(d-q^{r-1}a)}x_1^{d}\not\equiv x_1^d\mod I(P^1)$. It is not hard to check that we do not have the image of this monomial in $\D$, which implies that $\D$ is not Galois invariant. In the following example we show how this affects the bound for the minimum distance.
\end{rem}

\begin{ex}\label{excont}
We continue with Example \ref{exdualciclo}. We can consider $\Delta=\{0,1,2,3,4\}$, which gives $\Delta_\II=\II_0\cup \II_1$. However, we do not have $\abs{\II_{4}}=1$ and Proposition \ref{galoisinvariant} does not hold in this case. For instance, there are $t=2$ consecutive elements in $\Delta$, but the parameters of $\DT$ are $[17,15,2]$, and $2<t+1=3$. On the other hand, we have that $(\Delta')_\II=\II_0$, which only has $t=1$ consecutive elements, and Proposition \ref{bchmala} would give the parameters $[17,15,\geq 2]$.
\end{ex}

\section{Applications to EAQECCs}

This section is devoted to providing quantum codes from the linear codes developed in the previous section. Namely, we will construct EAQECCs using the CSS construction \cite[Thm. 4]{galindoentanglement} and the Hermitian construction  \cite[Thm. 3]{galindoentanglement}, as well as asymmetric EAQECCs \cite{galindoasymmetric}. 

\subsection{Euclidean EAQECCs}

In this section we will be interested in obtaining EAQECCs using the CSS construction \cite[Thm. 4]{galindoentanglement}. Given a nonempty set $U\subset \fq^n$, we denote by $\wt(U)$ the number $\min \{\wt(v)\mid v\in U\setminus \{0\}\}$, extending the notation that we have been using only for linear codes until now.

\begin{thm}[(CSS Construction)]\label{css}
Let $C_i\subset \fq^n$ be linear codes of dimension $k_i$, for $i=1,2$. Then, there is an EAQECC with parameters $[[n,\kappa,\delta;c]]_q$, where 
$$
\begin{aligned}
c&=k_1-\dim (C_1\cap C_2^\perp), \; \kappa=n-(k_1+k_2)+c, \text{ and } \\
\delta=&\min \left\{\wt\left(C_1^\perp\setminus \left(C_1^\perp\cap C_2 \right) \right), \wt\left(C_2^\perp\setminus \left(C_2^\perp\cap C_1 \right) \right) \right\}.
\end{aligned}
$$
\end{thm}

We are going to introduce some new notation for the codes we are going to use. In what follows, we are assuming that $p\mid N$.

\begin{defn}
Let $\mathcal{A}=\{a_0=0<a_1<\cdots< a_j\}$, the set of minimal representatives of the minimal cyclotomic sets. We are going to consider a set $\Delta=\bigcup_{i=0}^{t-1}\II_{a_i}\cup \{a_t\}$, i.e., the union of consecutive minimal cyclotomic sets with minimal representatives $a_0,\dots,a_{t-1}$, and the minimal element $a_t$. For such a set $\Delta$, we are going to consider the code $\mD$ defined as the linear code generated by $\{\ev_{X_N} (x_0 x_1^\alpha)\mid \alpha\in \Delta\setminus\{a_t\} \}\cup \{\ev_{X_N} (x_1^{a_t})\}$.
\end{defn}

\begin{rem}\label{remdeltaestrella}
If we look at the basis for the dual codes from Proposition \ref{propdual}, we see that $\mD=\PRS(N,\Delta^*)^\perp$, with $\Delta^*=\{0,1,\dots,N-1\}\setminus \bigcup_{i=0}^{t-1}\II_{N-1-a_i}$. In particular, the codes we are considering are not degenerate. 
\end{rem}

Although the previous remark shows that we can use the notation $\PRS(N,\Delta^*)^\perp$ instead of $\mD$, in what follows we are going to use $\mD$ because this will be the appropriate notation for Section \ref{sectraza}. This allows us to make reference to the following proofs directly from Section \ref{sectraza}, which helps to avoid repetition.

\begin{rem}\label{remarkdefs}
By the definitions, it is clear that $\mD=(\RS(N,\Delta'),0)+\langle \ev_{X_N}(x_1^{a_t})\rangle$, where $\Delta'=\Delta\setminus\{a_t\}$. This means that $\dim \mD=\dim \RS(N,\Delta')+1=\dim \RS(N,\Delta)$. We also have that $\dim \mDT=\dim \PRS(N,\Delta^*)_q=N+1-\sum_{i=0}^{t} n_{a_i}$ from Corollary \ref{dimprimario}. If $G_{N,\Delta}$ is a generator matrix of $\RS(N,\Delta)$, then we have that
\newcommand{\bigG}{\mbox{\normalfont\Large G}}
$$
\left(\begin{array}{c|c}
  G_{N,\Delta}
  & \begin{matrix}
   0 \\
   \vdots \\
   0
  \end{matrix} \\
\hline
  \ev_{Y_N}(x^{a_t}) & 1
\end{array}\right)
$$
is a generator matrix of $\mD$. We see that this does not correspond to any standard lengthening technique for linear codes. On the other hand, the BCH-type bound gives $\wt(\mDT)\geq \wt(\mD^\perp)\geq a_t+2$. 
\end{rem}

\begin{thm}\label{cuanticoeuclideo}
Let $\mathcal{A}=\{a_0=0<a_1<a_2<\cdots< a_z\}$ be the set of minimal representatives of the cyclotomic sets $\II_{a_i}$, $0\leq i\leq z$, of $\{0,1,\dots,N-1\}$ with respect to $q$. Let $\Delta=\bigcup_{i=0}^{t-1} \II_{a_i}\cup \{a_t\}$ such that $\RS(N,\Delta'')\subset \RS(N,\Delta'')^\perp$, where $\Delta''=\bigcup_{i=0}^{t} \II_{a_i}$. Then we can construct an EAQECC with parameters $[[n,\kappa,\geq \delta;c]]_{q}$, where $n=N+1$, $\kappa=N+1-2\left(\sum_{i=0}^t n_{a_i}\right) +c$, $\delta=a_t+2$, and $c\leq 1$. 
\end{thm}
\begin{proof}
We are going to consider the code $C_1=C_2=\mDTT$ for the CSS Construction \ref{css}. We have $\dim \mDTT=N+1-\dim \mDT=N+1-\dim \PRS(N,\Delta^*)_q=\sum_{i=0}^{t} n_{a_i}$ by Remark \ref{remarkdefs}.  Remark \ref{remarkdefs} also gives $\wt(\mDT)\geq  a_t+2$. 

For the parameter $c$, we claim that $\dim\left( \mDT\cap \mDTT\right)\geq \dim (\RS(N,\Delta'')_q,0) -1=\sum_{i=0}^t n_{a_i}-1$, which gives $c\leq 1$. Let $\Delta'=\Delta\setminus \{a_t\}$. By Remark \ref{remarkdefs} we have $\mD=(\RS(N,\Delta'),0)+\langle \ev_{X_N}(x_1^{a_t})\rangle$. 

We consider $v\in (\RS(N,\Delta)^\perp,0)$. Then $v$ is orthogonal to $(\RS(N,\Delta'),0)$ (taking into account that $\RS(N,\Delta')\subset \RS(N,\Delta)$), and it is also orthogonal to $\ev_{X_N}(x_1^{a_t})$ because the last coordinate of $v$ is 0, which means that $v\cdot \ev_{X_N}(x_1^{a_t})=v\cdot \ev_{X_N}(x_0x_1^{a_t})$, and $\ev_{X_N}(x_0x_1^{a_t})\in (\RS(N,\Delta),0)$. Therefore, $v\in \mD^\perp$. Taking into account the dimension and the fact that the codes $\mD^\perp$ are not degenerate, we can write $\mD^\perp=( \RS(N,\Delta)^\perp,0)+\langle w\rangle $, where $w$ is a vector with a nonzero last entry.

We consider a basis for $\mDT$ now, and we can also assume that all the vectors in the basis, besides one vector $w'$, have 0 as their last coordinate. Taking into account that $\mDT$ is not degenerate, this means that we have $\mDT=( (\RS(N,\Delta)^\perp)_q,0)+\langle w'\rangle $ for some vector $w'$ with nonzero last coordinate. In this case we have $(\RS(N,\Delta)^\perp)_q=(\RS(N,\Delta'')^\perp)_q$ because $\Delta^\perp$ and $\Delta''^\perp$ contain the same complete minimal cyclotomic sets (which is what matters in order to compute the subfield subcode of the dual, this can be seen using Theorem \ref{baseafin} and \cite[Prop. 3]{galindo1}). Moreover, we have that $\RS(N,\Delta'')_q\subset (\RS(N,\Delta'')^\perp)_q=(\RS(N,\Delta'')_q)^\perp$ because this code is Galois invariant by the reasoning after Corollary \ref{dimdual}.

Thus, we have seen that $\mDT\supset ((\RS(N,\Delta'')^\perp)_q,0)\supset (\RS(N,\Delta'')_q,0)$. On the other hand, we have
$$
\mDTT=\left((  \PRS(N,\Delta'')_q,0)+\langle (0,0,\dots,0,1)\rangle\right) \cap \langle w'\rangle^{\perp}.
$$
Note that $(0,0,\dots,0,1)\not \in \langle w'\rangle^{\perp}$ because $w'$ has a nonzero last coordinate. Hence, we can consider a basis for $\mDTT$ formed by $(\dim \RS(N,\Delta'')_q-1)$ vectors $u_i\in (\RS(N,\Delta'')_q,0)$, and a vector $w''$ such that its last coordinate is nonzero. Note that not all vectors can have the last coordinate equal to 0 because that would mean that we have the vector $(0,0,\dots,0,1)\in \mDT$, contradicting the bound given for the minimum distance. Therefore, all the vectors $u_i$ are in $\mDT\cap \mDTT$, which gives $c\leq 1$. 
\end{proof}

\begin{rem}\label{condicionesafin}
In \cite{galindostabilizer}, there are conditions in order to have $\RS(N,\Delta'')\subset \RS(N,\Delta'')^\perp$. For example, for the type of set $\Delta''$ that we are considering in Theorem \ref{cuanticoeuclideo}, if, for every cyclotomic set $\II_a\subset \Delta''$, we have $\II_{N-1-a}\not\subset \Delta''$, then $\RS(N,\Delta'')\subset \RS(N,\Delta'')^\perp$.
\end{rem}

For the code $\RS(N,\Delta'')^\perp$ we have the bound $\wt(\RS(N,\Delta'')^\perp)\geq a_{t+1}+1$. However, we have $a_{t+1}+1=a_t+2$ in many cases (this happens if and only if $a_t+1\not \in \Delta''$, because in that case $a_{t+1}=a_t+1$). In that situation, we have the same bound for the minimum distance for $\RS(N,\Delta'')^\perp$ and for the corresponding EAQECC from Theorem \ref{cuanticoeuclideo}. In the following discussion we will assume that $a_{t+1}+1=a_{t}+2$.

If we get a QECC with parameters $[[n,\kappa,\delta;0]]_q$ from the affine case using $\RS(N,\Delta'')_q$, then we would get an EAQECC with parameters $[[n+1,\kappa+1+c,\delta;c]]_q$ in the projective case using Theorem \ref{cuanticoeuclideo}, where $c\leq 1$. If we take into account the rate $\rho:=\kappa/n$ and the net rate $\overline{\rho}:=(\kappa-c)/n$, we see that the code obtained with Theorem \ref{cuanticoeuclideo} has better rate and net rate than the one obtained in the affine case. Moreover, it can be checked that the codes we obtain are not directly obtainable from the affine case using the propagation rules from \cite{grasslhowmuch}, which can be adapted for EAQECCs arising from Theorem \ref{css} (for example, see \cite{relativehull}).

In the constructions from Theorem \ref{cuanticohermitico} and Theorem \ref{paramquantraza}, the same argument shows that, as long as $a_{t+1}+1=a_{t}+2$, we can obtain codes with better rates than the ones from the affine case, which cannot be deduced from the propagation rules from \cite{grasslhowmuch}.

\begin{ex}\label{ejcuanteuclideo}
We consider $N=q^s=3^4=81$, with $q=3^2$ ($s=2$). The first minimal cyclotomic sets, ordered by their minimal element, are
$$
\begin{aligned}
\II_0=\{0\},\; \II_1=\{1,9\},\;\II_2=\{2,18\},\;\II_3=\{3,27\},\;\II_4=\{4,36\},\;
\II_5=\{5,45\}.
\end{aligned}
$$
With the notation that we have been using, we consider the minimal elements $a_i$, for $i=0,\dots,5$, and $\Delta=\bigcup_{i=0}^4\II_{a_i}\cup \{5\}$ ($t=5$ with the previous notation). We have $\sum_{i=0}^5n_{a_i}=11$, and we have $a_t+2=7$. It is easy to check that $\II_{N-1-{a_i}}\not \subset \Delta$ for $i=0,\dots,5$. By Remark \ref{condicionesafin} we have $\RS(N,\Delta'')\subset \RS(N,\Delta'')^\perp$ and we can apply Theorem \ref{cuanticoeuclideo} in order to obtain a quantum code with parameters $[[82,61,7;1]]_9$. If we had used the affine code $\RS(N,\Delta'')$ with $\Delta''=\bigcup_{i=0}^5\II_{a_i}$, the bound for the minimum distance would have been the same because $a_{t}+1=8\not\in \Delta''$, and we would have obtained the code $[[81,59,7;0]]_9$. 
\end{ex}

We can also get QECCs (EAQECCs with $c=0$) directly under some assumptions, as the following result shows.

\begin{prop}\label{quantgalois}
Assume that $p>2$. Let $N$ be an odd integer such that $N-1\mid q^s-1$ and $p\mid N$. We consider a union of cyclotomic sets $\Delta\subset \{0,1,\dots,N-1\}$ such that $d=d(\Delta)=(N-1)/2$. If $t$ is the number of consecutive exponents in $\Delta$, then we can construct a QECC with parameters $[[n,\kappa,\geq \delta;0]]_q$, where $n=N+1$, $\kappa=N+1-2\abs{\Delta}$, and $\delta=t+1$.
\end{prop}
\begin{proof}
By Proposition \ref{galoisinvariant}, Lemma \ref{lemaciclotomicotamano1} and Remark \ref{remdeltaestrella}, we have that $\PRS(N,\Delta)$ is Galois invariant, and we have $\wt((\PRS(N,\Delta)_q)^\perp)\geq t+1$. By Corollary \ref{dualfacil}, if we consider $\Delta_d=\{0,1,\dots,(N-1)/2\}$, we have that
$$
\D\subset \PRS(N,\Delta_d)=\PRS(N,\Delta_d)^\perp \subset \PRS(N,\Delta)^\perp.
$$
Therefore, considering the intersection with $\fq^n$ we obtain that $\PRS(N,\Delta)_q \subset (\PRS(N,\Delta)_q)^\perp$. If we consider $C_1=C_2=\PRS(N,\Delta)_q$ in the CSS Construction \ref{css}, we have already obtained the length, the bound for the minimum distance, and $c=0$, for the parameters of the corresponding quantum error-correcting code. For the dimension, we have $\dim \PRS(N,\Delta)_q=\abs{\Delta}$ by Corollary \ref{dimprimario}, taking into account that $\abs{\II_d}=1$ in this case by Lemma \ref{lemaciclotomicotamano1}. 
\end{proof}

\begin{ex}
We consider $q^s=3^3$, $q=3$ and $N=3^3=27$. Let $\Delta=\II_0\cup \II_1\cup \II_4\cup \II_{13}$ (note that $13=(N-1)/2$). We are not considering consecutive cyclotomic sets, which means that the BCH-type bound for the minimum distance might not be accurate. Hence, we have computed it directly with Magma \cite{magma}. The code $(\PRS(N,\Delta)_q)^\perp$ has parameters $[28,20,6]$, which gives a QECC with parameters $[[28,12,6;0]]_3$ by Proposition \ref{quantgalois}, which are the best known parameters for a quantum code over $\F_3$ with that length and dimension according to \cite{codetables}. With $\RS(N,\Delta)$ and $\RS(N,\Delta')$ (where $\Delta'=\Delta\setminus \{(N-1)/2\}$), we obtain the parameters $[27,19,6]$ and $[27,20,5]$ for the dual codes of their subfield subcodes, respectively. These codes would give QECCs with parameters $[[27,11,6;0]]_3$ and $[[27,13,5;0]]_3$, respectively, applying the CSS Construction \ref{css}.
\end{ex}

\subsection{Asymmetric EAQECCs}

As we said in the introduction, phase-shift and qudit-flip errors are not equally likely to occur. It is therefore desirable to obtain EAQECCs with different error correction capabilities for each of these types of errors. In order to construct asymmetric EAQECCs, we can use the following result from \cite{galindoasymmetric}.

\begin{thm}\label{asimetricos}
Let $C_i\subset \fq^n$ be linear codes of dimension $k_i$, for $i=1,2$. Then, there is an asymmetric EAQECC with parameters $[[n,\kappa,\delta_z/\delta_x;c]]_q$, where
$$
\begin{aligned}
c&=k_1-\dim (C_1\cap C_2^\perp), \; \kappa=n-(k_1+k_2)+c, \\
\delta_z=\wt&\left(C_1^\perp\setminus \left(C_1^\perp\cap C_2 \right)\right) \text{ and }\; \delta_x=  \wt\left(C_2^\perp\setminus \left(C_2^\perp\cap C_1 \right) \right) .
\end{aligned}
$$
\end{thm}

The two minimum distances $\delta_z$ and $\delta_x$ give the error correction capability of the corresponding asymmetric EAQECC, which can correct up to $\lfloor(\delta_z-1)/2\rfloor$ phase-shift errors and $\lfloor(\delta_x-1)/2\rfloor$ qudit-flip errors.

In sections \ref{secprimario} and \ref{secdual} we obtained bases for both the primary codes $\DS$ and their duals $\DT$. This is the key for the proof of the following result, which allows us to construct asymmetric EAQECCs from subfield subcodes of projective Reed-Solomon codes. We recall that, for $\Delta\subset \{0,1,\dots,N-1\}$, we denote $\Delta_\II=\bigcup_{\II_a\subset \Delta}\II_a$, and we also recall that $\mathcal{B}$ is the set of maximal representatives of the minimal cyclotomic sets.

\begin{thm}\label{asimetricosproyectivos}
Let $1\leq d_1,d_2 \leq N-1$, such that $d_i\in \mathcal{B}$, for $i=1,2$, and $p\mid N$. We consider $\Delta_{d_i}=\{0,1,\dots,d_i\}$ and we denote $\Delta_{d_i}':=\Delta_{d_i}\setminus \{d_i\}$, for $i=1,2$. If $((\Delta'_{d_1})_\II)^\perp\subset (\Delta_{d_2}')_\II$, then we can construct an asymmetric EAQECC with parameters 
$$
[[N+1,\sum_{b\in \mathcal{B},b<d_1}{n_b}+\sum_{b\in \mathcal{B},b<d_2}{n_b}+2-N,\delta_z/\delta_x;1]]_q,
$$
where $\delta_z\geq N-d_1+1$, $\delta_x\geq N-d_2+1$.
\end{thm}
\begin{proof}
We are going to consider $C_i=(\PRS(N,\Delta_{d_i})_q)^\perp$, for $i=1,2$, and we are going to use Theorem \ref{asimetricos}. The bounds for $\delta_z$ and $\delta_x$ are clear, and we obtain the dimension using Corollary \ref{dimprimario} if we assume $c=1$. For the parameter $c=\dim (\PRS(N,\Delta_{d_1})_q)^\perp-\dim((\PRS(N,\Delta_{d_1})_q)^\perp\cap \PRS(N,\Delta_{d_2})_q)$, we are going to study $\dim((\PRS(N,\Delta_{d_1})_q)^\perp\cap \PRS(N,\Delta_{d_2})_q)$. For $(\PRS(N,\Delta_{d_1})_q)^\perp$ we have the basis given by the evaluation of the following set from Theorem \ref{baseproyectivadual}:
\begin{equation}\label{conj1}
\bigcup_{a\in \mathcal{A}\mid \II_a\cap \Delta_{d_1}^\perp \neq \emptyset}\{\mathcal{T}_a(\xi_a^r x_0 x_1^a)\mid 0\leq r\leq n_a-1 \}\cup \{ \mathcal{T}_{N-1-d_1}(\xi_{N-1-d_1}^r x_1^{N-1-d_1}) \mid 0\leq r \leq n_{d_1}-1\}.
\end{equation}
From Theorem \ref{baseproyectiva} it is easy to obtain that the evaluation of the following set gives a basis for $\PRS(N,\Delta_{d_2})_q$:
\begin{equation}\label{conj2}
\bigcup_{a\in \mathcal{A}\mid \II_a\subset \Delta'_{d_2}}\{\mathcal{T}_a(\xi_a^r x_0x_1^a)\mid 0\leq r\leq n_a-1 \} \cup \{\mathcal{T}_{d_2}^h(x_1^{d_2}) \}.
\end{equation}
It is also clear that the $a\in \mathcal{A}$ such that $\II_a\cap \Delta_{d_1}^\perp \neq \emptyset$ are precisely the $a\in \mathcal{A}$ such that $\II_a\subset ((\Delta_{d_1})_\II)^\perp$. We also have that $((\Delta'_{d_1})_\II)^\perp=((\Delta_{d_1})_\II)^\perp\cup \II_{N-1-d_1}$. Therefore, taking into account the assumption $((\Delta'_{d_1})_\II)^\perp\subset (\Delta_{d_2}')_\II\subset \Delta'_{d_2}$, we have that all the traces of monomials of the type $x_0x_1^{a}$, with $a\in\mathcal{A}$, in the set from (\ref{conj1}), are contained in the set from (\ref{conj2}). This implies that the evaluation of the set 
\begin{equation}\label{conj3}
\bigcup_{a\in \mathcal{A}\mid \II_a\cap \Delta_{d_1}^\perp \neq \emptyset}\{\mathcal{T}_a(\xi_a^r x_0 x_1^a)\mid 0\leq r\leq n_a-1 \}
\end{equation}
is in $(\PRS(N,\Delta_{d_1})_q)^\perp\cap \PRS(N,\Delta_{d_2})_q$.

Now we are going to study which polynomials from the set generated by
$$
\{ \mathcal{T}_{N-1-d_1}(\xi_{N-1-d_1}^r x_1^{N-1-d_1}) \mid 0\leq r \leq n_{d_1}-1\}
$$
have their evaluation in $(\PRS(N,\Delta_{d_1})_q)^\perp\cap \PRS(N,\Delta_{d_2})_q$. As in Theorem \ref{baseproyectivadual}, we assume that $\xi_{N-1-d_1}$ is a primitive element of $\F_{q^{n_{d_1}}}$ (note that $n_{d_1}=n_{N-1-d_1}$) such that $\mathcal{T}_{N-1-d_1}(\xi_{N-1-d_1})\neq 0$. For ease of notation, we are going to denote now $d'_1=N-1-d_1$. For $0\leq r\leq n_{d_1}-1$, $r\neq 1$, we have
\begin{equation}\label{polint}
\begin{aligned}
\mathcal{T}_{d'_1}(\xi_{d'_1})\mathcal{T}_{d'_1}(\xi_{d'_1}^r x_0x_1^{d'_1})-\mathcal{T}_{d'_1}(\xi^r_{d'_1})&\mathcal{T}_{d'_1}(\xi_{d'_1}x_0x_1^{d'_1}) \\
&\equiv \mathcal{T}_{d'_1}(\xi_{d'_1})\mathcal{T}_{d'_1}(\xi_{d'_1}^r x_1^{d'_1})-\mathcal{T}_{d'_1}(\xi^r_{d'_1})\mathcal{T}_{d'_1}(\xi_{d'_1}x_1^{d'_1}) \bmod I(X_N).
\end{aligned}
\end{equation}
This is easy to see because when we set $x_0=1$, we obtain the same polynomials at each side, which means that they have the same evaluation in $[\{1\}\times Y_N]$, and both polynomials evaluate to $0$ in $[0:1]$. Therefore, they have the same evaluation in $X_N$. Because of the assumption $((\Delta'_{d_1})_\II)^\perp=((\Delta_{d_1})_\II)^\perp\cup \II_{d'_1}\subset (\Delta_{d_2}')_\II$, we obtain $\II_{d_1'}\subset \Delta'_{d_2}$ and it is clear that we have the evaluation of the polynomial in the left-hand side of (\ref{polint}) in $\PRS(N,\Delta_{d_2})_q$ if we consider the basis from (\ref{conj2}). The evaluation of the polynomial in the right-hand side is clearly in $(\PRS(N,\Delta_{d_1})_q)^\perp$ (see (\ref{conj1})). Thus, we have proved that the image by the evaluation map of the polynomials in the set 
\begin{equation}\label{conj4}
\{\mathcal{T}_{d'_1}(\xi_{d'_1})\mathcal{T}_{d'_1}(\xi_{d'_1}^r x_0x_1^{d'_1})-\mathcal{T}_{d'_1}(\xi^r_{d'_1})\mathcal{T}_{d'_1}(\xi_{d'_1}x_0x_1^{d'_1})\mid 0\leq r\leq n_{d_1}-1,r\neq 1\}
\end{equation}
is in $(\PRS(N,\Delta_{d_1})_q)^\perp\cap \PRS(N,\Delta_{d_2})_q$. 

Hence, the evaluation of the union of the sets from (\ref{conj3}) and (\ref{conj4}) is in $(\PRS(N,\Delta_{d_1})_q)^\perp\cap \PRS(N,\Delta_{d_2})_q$, and it is easy to see that the evaluation of this union is linearly independent. Taking into account the basis from (\ref{conj1}), we obtain that $\dim((\PRS(N,\Delta_{d_1})_q)^\perp\cap \PRS(N,\Delta_{d_2})_q)\geq \dim((\PRS(N,\Delta_{d_1})_q)^\perp)-1$, i.e., $c\leq 1$. 

On the other hand, having $c=0$ means that $(\PRS(N,\Delta_{d_1})_q)^\perp\subset \PRS(N,\Delta_{d_2})_q$. This implies that the evaluation of all the traces appearing in (\ref{polint}) are in $\PRS(N,\Delta_{d_2})_q$. However, the evaluations of $\mathcal{T}_{d'_1}(\xi_{d'_1}x_0x_1^{d'_1})$ and $\mathcal{T}_{d'_1}(\xi_{d'_1}x_1^{d'_1})$ differ only at the coordinate associated to the point $[0:1]$. This would imply that the minimum distance of $ \PRS(N,\Delta_{d_2})_q$ is 1, a contradiction. Therefore, $c=1$.
\end{proof}

\begin{rem}
We note that in the previous result we have that $(\Delta'_{d})_\II=\bigcup_{b\in \mathcal{B}\mid b< d}\II_b$.
\end{rem}

As we said in the introduction, it is desirable to obtain asymmetric quantum codes with higher error-correction capability for phase-shift errors, i.e. with $\delta_z>\delta_x$. For the codes obtained using Theorem \ref{asimetricosproyectivos}, this corresponds to choosing $d_1<d_2$.

In the next example we show that we are able to obtain codes which are better than the ones available in the current literature.

\begin{ex}\label{exasymmetric}
We consider the extension $\F_{16}\supset \F_{4}$, which is the setting from Example \ref{exdualciclo}. We choose $d_1=14$, $d_2=15$, and apply Theorem \ref{asimetricosproyectivos}, which gives the parameters $[[17,14,3/2;1]]_4$. In \cite{galindoasymmetric}, we can find a code with parameters $[[15,12,3/2;1]]_4$ using BCH codes. We see that the code we have obtained has better rate $\kappa/n$, and also better net rate $(\kappa-c)/n$.

If we consider the extension $\F_{25}\supset \F_{5}$ instead, and choose $d_1=22$, $d_2=23$, we obtain a code with parameters $[[26,19,4/3;1]]_5$ using Theorem \ref{asimetricosproyectivos}. It is possible to adapt the propagation rules from \cite{grasslhowmuch} to asymmetric EAQECCs arising from Theorem \ref{asimetricos}. For example, we can reduce the length by using extra entanglement, provided that $c\leq n-\kappa-2$:
\begin{equation}\label{propruleasym}
[[n,\kappa,\delta_z/\delta_x;c]]_q\rightarrow[[n-1,\kappa,\delta_z/\delta_x;c+1]]_q. 
\end{equation}
In \cite{galindoasymmetric} a code with parameters $[[24,19,4/3;3]]_5$ is presented, which can be obtained from our code with parameters $[[26,19,4/3;1]]_5$ by applying (\ref{propruleasym}) two times. In this sense, we can say that the parameters $[[24,19,4/3;3]]_5$ appearing in \cite{galindoasymmetric} are a consequence of the parameters $[[26,19,4/3;1]]_5$ that we obtain with Theorem \ref{asimetricosproyectivos}. 

Finally, if we consider the extension $\F_{64}\supset \F_8$, for $d_1=60$ and $d_2=63$, we obtain the parameters $[[65,58,5/2;1]]_8$, which give a better rate and net rate than the code with parameters $[[63,56,5/2;1]]_8$ from \cite{galindoasymmetric}. If we choose $d_1=58$ and $d_2=62$ instead, we obtain the parameters $[[65,52,7/3;1]]_8$, which, after using the propagation rule (\ref{propruleasym}) as before, give the parameters $[[63,52,7/3;3]]_8$ that appear in \cite{galindoasymmetric}.
\end{ex}

If we do not assume $((\Delta'_{d_1})_\II)^\perp\subset (\Delta_{d_2}')_\II$ in Theorem \ref{asimetricosproyectivos}, then we would obtain instead the parameters $[[N+1,\sum_{b\in \mathcal{B},b<d_1}{n_b}+\sum_{b\in \mathcal{B},b<d_2}{n_b}+1+c-N,\delta_z/\delta_x;c]]_q$, for $c=\dim (\PRS(N,\Delta_{d_1})_q)^\perp-\dim((\PRS(N,\Delta_{d_1})_q)^\perp\cap \PRS(N,\Delta_{d_2})_q)$.

\subsection{Hermitian EAQECCs}

In the Hermitian case, we have to work with three different fields. Hence, we are going to change the notation from the previous sections. We consider the field extension $\F_{q^{2\ell}}\supset \F_{q^2}$, where $q^{2\ell}=p^{2r}$, $q=p^s$, for some $r,s>0$, and $r=\ell s$. Thus, in what follows we are going to obtain codes of length $n=N+1$, where $N>1$ is an integer such that $N-1\mid q^{2\ell}-1$. 

As before, we are going to consider the set $\mathbb{Z}_{N}=\{0\}\cup \{1,2,\dots,N-1\}$, where $\{1,2,\dots,N-1\}$ is regarded as the set of representatives of the ring $\mathbb{Z}/(N-1)\mathbb{Z}$. We consider the cyclotomic sets with respect to $q^2$ over $\{0,1,\dots,N-1\}$. We call $\mathcal{A}$ the set of minimal elements of the different cyclotomic sets. We introduce now the Hermitian construction \cite[Thm. 3]{galindoentanglement} that we are going to use.

\begin{thm}[(Hermitian construction)]\label{hermitica}
Let $C\subset \F_{q^2}^n$ be a linear code of dimension $k$ and $C^{\perp_h}$ its Hermitian dual. Then, there is an EAQECC with parameters $[[n,\kappa,\delta;c]]_q$, where 
$$
c=k-\dim(C\cap C^{\perp_h}), \;\kappa=n-2k+c, \; \text{ and } \;\delta=\wt(C^{\perp_h}\setminus (C\cap C^{\perp_h})).
$$
\end{thm}

We are only going to consider the Hermitian product over $\F_{q^2}$. Therefore, for $a,b\in \F_{q^{2}}^{n}$ we have 
$$
a\cdot_{h} b:=\sum_{i=0}^n a_i b_i^{q}.
$$

In what follows, when considering a power of a code or a vector, we will be considering the component-wise power, i.e., $C^{q}:=\{c^{q}:=(c_1^{q},\dots,c_n^{q})\mid c=(c_1,\dots,c_n)\in C\}$. It is easy to check that, for codes over $\F_{q^2}$, we have that $C^\perp=(C^{\perp_h})^q$, where $C^{\perp_h}$ denotes the Hermitian dual.

\begin{thm}\label{cuanticohermitico}
Let $\mathcal{A}=\{a_0=0<a_1<a_2<\cdots< a_z\}$ be the set of minimal representatives of the cyclotomic sets $\II_{a_i}$, $0\leq i\leq z$, of $\{0,1,\dots,N-1\}$ with respect to $q^2$. Let $\Delta=\bigcup_{i=0}^{t-1} \II_{a_i}\cup \{a_t\}$ such that $\RS(N,\Delta'')_{q^2}\subset (\RS(N,\Delta'')_{q^2})^{\perp_h}$, where $\Delta''=\bigcup_{i=0}^{t} \II_{a_i}$. Then we can construct an EAQECC with parameters $[[n,\kappa,\geq \delta;c]]_{q}$, where $n=N+1$, $\kappa=N+1-2\left(\sum_{i=0}^t n_{a_i}\right) +c$, $\delta=a_t+2$ and $c\leq 1$. 
\end{thm}

\begin{proof}
We are going to consider the code $C=\mDTTh$ for the Hermitian construction \ref{hermitica}. Using what we obtained in Theorem \ref{cuanticoeuclideo}, the only thing left to prove is the statement about the parameter $c$.

Following the proof of Theorem \ref{cuanticoeuclideo}, we have $\mDTq=( (\RS(N,\Delta'')^\perp)_{q^2},0)+\langle w'\rangle $ for some vector $w'$ with nonzero last coordinate. Therefore, we see that $\dim \mDTTh = \dim \RS(N,\Delta'')_{q^2} = \dim ((\RS(N,\Delta'')^{\perp})_{q^2})^{\perp_{h}}$. Moreover, we have
$$
((\RS(N,\Delta'')^{\perp})_{q^2})^{\perp_{h}}=((\RS(N,\Delta'')_{q^2})^\perp)^{\perp_{h}}=(((\RS(N,\Delta'')_{q^2})^\perp)^{\perp})^q=\left(\RS(N,\Delta'')_{q^2}\right)^q.
$$

Thus, we obtain 
$$
\mDTTh=(( ( \PRS(N,\Delta'')_{q^2})^q,0)+\langle (0,0,\dots,0,1)\rangle) \cap \langle (w')\rangle^{\perp_h}.
$$
Note that $(0,0,\dots,0,1)\not \in \langle (w')\rangle^{\perp_h}$ because $w'$ has a nonzero last coordinate. We can consider a basis for $\mDTTh$ formed by $(\dim \RS(N,\Delta'')_{q^2}-1)$ vectors $u_i\in ((\RS(N,\Delta'')_{q^2})^{q},0)$, and a vector $w$ such that its last coordinate is nonzero (not all vectors can have the last coordinate equal to 0 because that would mean that we have the vector $(0,0,\dots,0,1)\in \mDTh$, contradicting the bound given for the minimum distance).

By our hypothesis, we have $\RS(N,\Delta'')_{q^2}\subset (\RS(N,\Delta'')_{q^2})^{\perp_h}$. This implies that $\left(\RS(N,\Delta'')_{q^2}\right)^{q}\subset \left((\RS(N,\Delta'')_{q^2})^{\perp_h}\right)^{q}=(\RS(N,\Delta'')_{q^2})^\perp=(\RS(N,\Delta'')^\perp)_{q^2}$. Taking into account that $\mDTh \supset ((\RS(N,\Delta'')^\perp)_{q^2},0)\supset ((\RS(N,\Delta'')_{q^2})^{q},0)$, we see that the vectors $u_i$ are in $\mDTh$ as well, and we obtain the desired inequality for the dimension of the intersection.
\end{proof}

\begin{rem}\label{remhermitico}
From \cite[Prop. 3]{galindostabilizer} we can obtain conditions in order to have $\RS(N,\Delta'')_{q^2}\subset (\RS(N,\Delta'')_{q^2})^{\perp_h}$, like the one we show next. Let $\Delta''=\bigcup_{i=0}^{t} \II_{a_i}$, and we denote by $a'_i$ the minimal element in $\mathcal{A}$ such that $\II_{a'_i}=\II_{-qa_i}$. Assuming $d(\Delta)<N-1$, if $\Delta''\subset (\Delta'')^{\perp_h}:=\{0,1,\dots,N-1\}\setminus \bigcup_{i=0}^t \II_{a'_i}$, then we have $\RS(N,\Delta'')_{q^2}\subset (\RS(N,\Delta'')_{q^2})^{\perp_h}$.
\end{rem}

\begin{ex}
We continue with the setting from Example \ref{ejcuanteuclideo}. It is easy to check that the set $\Delta$ in Example \ref{ejcuanteuclideo} satisfies $\Delta\subset \Delta^{\perp_h}$, and by Remark \ref{remhermitico} and Theorem \ref{cuanticohermitico} we obtain a quantum code with parameters $[[82,67,7;1]]_3$.
\end{ex}

With the construction from Theorem \ref{cuanticohermitico} we can obtain several quantum codes over $\F_2$ whose parameters do not appear in the table of EAQECCs from \cite{codetables}, and therefore we improve the table. With the extension $\F_{2^4}\supset \F_{2^2}$, we can obtain a code with parameters $[[17,12,3;1]]_2$, which is not in the table from \cite{codetables}. We can consider now the following propagation rule from \cite{grasslhowmuch}: let $C$ be an EAQECC with parameters $[[n,\kappa,\delta;c]]_q$ obtained from Theorem \ref{hermitica} (for example, the codes from Theorem \ref{cuanticohermitico}). If $c\leq n-\kappa-2$, then we can reduce the length by using extra entanglement:
\begin{equation}\label{proprule2}
    [[n,\kappa,\delta;c]]_q\rightarrow [[n-1,\kappa,\delta;c+1]]_q.
\end{equation}
Iterating this rule, it is easy to check that, from an EAQECC with parameters $[[n,\kappa,\delta;c]]_q$, one can obtain EAQECCs with parameters $[[n-s,\kappa,\delta;c+s]]_q$, $s=1,\dots,(n-\kappa-c)/2$. Note that the maximum value for $c$ is $k=\dim C$, where $C$ is the classical code used for Theorem \ref{hermitica}, and for the maximum value of $s$ that we have stated we have precisely that $c+s=k$:
$$
c+s=c+\frac{n-\kappa-c}{2}=c+\frac{2k-2c}{2}=k.
$$

Applying the propagation rule (\ref{proprule2}) to the parameters $[[17,12,3;1]]_2$, we obtain $[[16,12,3;2]]_2$ and $[[15,12,3;3]]_2$, which are also missing in the table \cite{codetables}.

For the extension $\F_{2^6}\supset \F_{2^2}$, we obtain codes with length $65$, which is greater than the current maximum length in \cite{codetables} for EAQECCs over $\F_2$. Nevertheless, we can reduce the length with the propagation rule (\ref{proprule2}) and check if the corresponding parameters are in the table. A code with parameters $[[64,58,3;2]]_2$, whose parameters are missing in \cite{codetables}, is obtained from the code with parameters $[[65,58,3;1]]_2$ derived from Theorem \ref{cuanticohermitico} using (\ref{proprule2}). Moreover, by applying the propagation rule (\ref{proprule2}) to the code with parameters $[[65,40,7;1]]_2$ deduced from Theorem \ref{cuanticohermitico}, we obtain codes with parameters $[[65-i,40,7;1+i]]_2$, for $i=1,2,\dots,12$, whose parameters are also missing in \cite{codetables}.

In total, we obtain in this way 16 EAQECCs over $\F_2$ whose parameters are missing in \cite{codetables}.

The table of EAQECCs from \cite{codetables} also covers codes over $\F_3$. However, the smaller length that we can achieve with Theorem \ref{cuanticohermitico} over $\F_3$ would be $3^4+1=82$, much higher than the current maximum length in the table from \cite{codetables} for this case. For example, we obtain codes with parameters $[[82,77,3;1]]_3$, $[[82,73,4;1]]_3$, $[[82,69,5;1]]_3$ and $[[82,65,6;1]]_3$.

\section{Evaluating at the trace roots}\label{sectraza}

In this section, following the ideas from \cite{fernandotrace}, we are going to consider evaluation codes over the roots of a suitable trace polynomial. In \cite{fernandotrace}, the authors considered the trace polynomial over $\F_{q^{2\ell}}$ with respect to $\fq$ defined as
$$
\Trr(x)=x+x^q+x^{q^2}+\cdots + x^{q^{2\ell-1}}.
$$

Let $Y_{\Tr_\ell}=\{\alpha\in \F_{q^{2\ell}}\mid \Trr(\alpha)=0 \}$. It is well known that $\abs{Y_{\Tr_\ell}}=q^{2\ell-1}$. In \cite{fernandotrace}, evaluation codes over the roots of the trace are defined, obtaining codes with length $q^{2\ell-1}$, and bounds for the dimension and minimum distance of these codes are found. In this section we are going to do a similar thing over the projective space, obtaining codes of length $q^{2\ell-1}+1$.

Firstly, we need to define the finite set of projective points in which we are going to evaluate. In order to do this, we are simply going to add the point at infinity to the set of roots of the trace, i.e., we are going to consider the following set of points:
$$
\mathbb{X}_{\Trr}=\{ [1:\alpha] \mid \Trr(\alpha)=0\}\cup \{[0:1]\}.
$$

It is clear from the definition that $\abs{\Xr}=q^{2\ell-1}+1$. Moreover, we can give this set as the zeroes of a square-free homogeneous polynomial. In the rest of this section, when we consider the homogenization $f^h$ of a polynomial $f$, we are considering the standard homogenization (up to degree $\deg(f)$).

\begin{prop}
The vanishing ideal of $\Xr$ is $I(\Xr)=\langle x_0(\Trr(x_1))^h \rangle$.
\end{prop}
\begin{proof}
The generator of the ideal is a homogeneous polynomial. Therefore, we can just look at the set of representatives $P^1$ to check the zeroes of the ideal. It is clear that $[0:1]\in V(\langle x_0(\Trr(x_1))^h \rangle)$. And it is also clear that if $[1:\alpha]$ is a zero of $x_0(\Trr(x_1))^h$, then $\alpha$ must be a root of $\Trr(x)$. Thus, we have that $V(\langle x_0(\Trr(x_1))^h \rangle)=\Xr$. 

On the other hand, we have the decomposition
$$
\Trr(x)=\prod_{\alpha\in \F_{q^{2\ell}}\mid \Trr(\alpha)=0}(x-\alpha).
$$

Homogenizing and multiplying by $x_0$ we get
$$
x_0(\Trr(x_1))^h=x_0\prod_{\alpha\in \F_{q^{2\ell}}\mid \Trr(\alpha)=0}(x_1-\alpha x_0).
$$
Therefore, $x_0(\Trr(x_1))^h$ is a square-free polynomial and $\langle x_0(\Trr(x_1))^h \rangle$ is a radical ideal by \cite[Prop. 9, Chapter 4, Section 2]{cox}, which means that it is equal to $I(\Xr)$.
\end{proof}

If we consider the set of standard representatives $\xr$ of $\Xr$, we obtain the following vanishing ideal.

\begin{prop}\label{vanishingtrazaafin}
The vanishing ideal of $\xr$ is
$$
I(\xr)=\langle x_0^2-x_0,x_1^{q^{2\ell}}-x_1,(x_0-1)(x_1-1),x_0\Tr(x_1)\rangle .
$$
\end{prop}
\begin{proof}
It is clear that any point of $\xr$ satisfies the equations. On the other hand, any point that satisfies this equations must have the first coordinate equal to 0 or 1 because of the first equation. If it is 0, then by the equation $(x_0-1)(x_1-1)\equiv 0 \bmod I(\xr)$ we have that the last coordinate is equal to 1. If the first coordinate is 1, then the last equation implies that the last coordinate must be a zero of $\Tr(x)$. Therefore, $V(I(\xr))=\xr$. We obtain the result applying Seidenberg's Lemma \cite[Prop. 3.7.15]{kreuzer1} and Hilbert's Nullstellensatz over the algebraic closure of $\F_{q^{2\ell}}$.
\end{proof}

We are going to define the evaluation map that we are going to use in order to construct these new codes (we have $n=q^{2\ell-1}+1$): 
$$
\ev_{\Trr}:\F_{q^{2\ell}}[x_0,x_1]/I(\xr) \rightarrow \F_{q^{2\ell}}^{n},\:\: f\mapsto \left(f(P_1),\dots,f(P_{n})\right)_{P_i \in \xr}.
$$

\begin{defn}
Let $\mathcal{A}=\{a_0=0<a_1<\cdots< a_z\}$. We are going to consider a set $\Delta=\bigcup_{i=0}^{t-1}\II_{a_i}\cup \{a_t\}$ as before. For such a set $\Delta$, we consider the code $\Dr$ defined as the linear code generated by $\{\ev_{\Trr} (x_0 x_1^\alpha)\mid \alpha\in \Delta\setminus\{a_t\} \}\cup \{\ev_{\Trr} (x_1^{a_t})\}$.
\end{defn}

In what follows we are going to need to use the codes $\RS(\Trr,\Delta):=\RS(Y_{\Trr},\Delta)$ that appear in \cite{fernandotrace}, which are the puncturing of the codes $\Dr$ at the coordinate associated to the point $[0:1]$. When $\Delta$ is a union of consecutive cyclotomic sets, we have that $(\RS(\Trr,\Delta)_{q^2})^\perp=(\RS(\Trr,\Delta)^\perp)_{q^2}$. We are going to be interested in the code $\Drts$, for which we have the following result.

\begin{thm}\label{thmtraza}
Let $a_0=0<a_1<a_2<\cdots<a_{t-1}<a_{t}<q^{2\ell}-1$ be a sequence of consecutive elements of $\mathcal{A}$. Let $\Delta=\bigcup_{i=0}^{t-1} \II_{a_i}\cup \{a_t\}$ and let $\Delta''=\Delta\cup \II_{a_t}$. Assuming that $\Drts$ is not degenerate, we have the following inequalities:
$$
\begin{aligned}
\dim \Drts &=  \dim (\RS(\Trr,\Delta'')^\perp)_{q^2}+1\geq n-\sum_{i=0}^{t} n_{a_i},\\
\wt(\Drts)&\geq a_t+2.
\end{aligned}
$$
\end{thm}
\begin{proof}
By the definitions, it is clear that $\Dr=(\RS(\Trr,\Delta'),0)+\langle \ev_{\Trr}(x_1^{a_t})\rangle$, where $\Delta'=\Delta\setminus\{a_t\}$. This means that $\dim \Dr=\dim \RS(\Trr,\Delta')+1=\dim \RS(\Trr,\Delta)$, because if we have $\dim \RS(\Trr,\Delta')=\dim \RS(\Trr,\Delta)$, this means that $\ev_{\Trr}(x_0x_1^{a_t})$ is in $(\RS(\Trr,\Delta'),0)$, which implies that $\ev_{\Trr}(x_0x_1^{a_t}-x_1^{a_t})$ is in $\Dr$, but this is a vector of weight 1, which is a contradiction because $\Drts$ (and $\Drt$) is not degenerate. Therefore, we have that $\dim \Drt=\dim \RS(\Trr,\Delta)^\perp+1$.

Arguing as in the proof of Theorem \ref{cuanticoeuclideo}, we have $\Drt=( \RS(\Trr,\Delta)^\perp,0)+\langle w\rangle $, where $w$ is a vector with a nonzero last entry, and we also obtain $\Drts=( (\RS(\Trr,\Delta)^\perp)_{q^2},0)+\langle w'\rangle $ for some vector $w'$ with nonzero last coordinate. Moreover, a basis for $( (\RS(\Trr,\Delta)^\perp)_{q^2},0)$ would give us $\dim (\RS(\Trr,\Delta)^\perp)_{q^2}$ linearly independent vectors with last coordinate equal to 0, which means that $\dim \Drts=\dim (\RS(\Trr,\Delta)^\perp)_{q^2}+1$. 

We obtain $\dim \Drts=\dim (\RS(\Trr,\Delta'')^\perp)_{q^2}+1$ because $(\RS(\Trr,\Delta)^\perp)_{q^2}=(\RS(\Trr,\Delta'')^\perp)_{q^2}$, which is what we are going to see next. When evaluating in all the points of $\F_{q^{2\ell}}$, we have $(\RS(q^{2\ell},\Delta)^\perp)_{q^2}=(\RS(q^{2\ell},\Delta'')^\perp)_{q^2}$. The code $\RS(\Trr,\Delta)$ (resp. $\RS(\Trr,\Delta'')$) corresponds to a puncturing of $\RS(q^{2\ell},\Delta)$ (resp. $\RS(q^{2\ell},\Delta'')$) because we only evaluate in the zeroes of $\Trr(x)$. The dual of a punctured code is equal to the shortening of the dual code at the same positions \cite[Prop. 2.1.17]{pellikaanlibro}. Given a code $C$, if we denote by $S$ the positions where we are puncturing (resp. shortening), by $C_S$ the punctured code and by $C^S$ the shortened code, we obtain
$$
((C_S)^\perp)_{q^2}=((C^{\perp})^S)_{q^2}=((C^\perp)_{q^2})^S,
$$
because shortening a code commutes with considering its subfield subcode. Let $S$ be the positions where we puncture in order to obtain $\RS(\Trr,\Delta)$ from $\RS(q^{2\ell},\Delta)$. Using the previous expression and the fact that $(\RS(q^{2\ell},\Delta)^\perp)_{q^2}=(\RS(q^{2\ell},\Delta'')^\perp)_{q^2}$ we get
$$
\begin{aligned}
(\RS(\Trr,\Delta)^\perp)_{q^2}&=((\RS(q^{2\ell},\Delta)_S)^\perp)_{q^2}=((\RS(q^{2\ell},\Delta)^\perp)_{q^2})^S= ((\RS(q^{2\ell},\Delta'')^\perp)_{q^2})^S\\
&=((\RS(q^{2\ell},\Delta'')_S)^\perp)_{q^2}=(\RS(\Trr,\Delta'')^\perp)_{q^2}.
\end{aligned}
$$

The bound for the dimension given in the statement is obtained by using \cite[Thm. 13]{fernandotrace}.

On the other hand, for the minimum distance, we have the BCH-type bound for $\Drt$, which gives $\wt(\Drt)\geq a_t+2$, and it is inherited by $\Drts$.
\end{proof}

The previous result shows that, if $a_{t+1}+1=a_t+2$, then the code $\Drts$ has 1 more length and dimension than the code $ (\RS(\Trr,\Delta'')^\perp)_{q^2}$. In the next example we obtain some codes $\Drts$ with record parameters according to \cite{codetables}.

\begin{ex}\label{ejemplocerostraza}
We consider the field extension $\F_{2^8}\supset \F_{2^2}$, i.e., we have $q=2$ and $\ell=4$. Therefore, we will get codes with length $N=129$. Let $\Delta=\II_{0}\cup \II_{1}\cup \cdots \II_{a_{t-1}}\cup \{a_t\}$. Hence, we have $\wt\left(\Drts\right)\geq a_t+2.$ The dimension of these codes can be easily computed using Magma \cite{magma}. In this case, we obtain a lot of codes whose parameters achieve the best known values in \cite{codetables}, and in many cases we are obtaining codes with higher length and dimension, but same minimum distance as in the affine case. Moreover, we obtain the parameters $[129,90,15]_4$, $[129,86,16]_4$ and $[129,41,44]_4$, for $a_t$ equal to $13,14$ and $42$, respectively. In \cite{codetables}, a construction of a code with parameters $[129,86,16]_4$ is currently missing, and we are able to obtain one. The codes with parameters $[129,90,15]_4$ and $[129,41,44]_4$ exceed the best known values in \cite{codetables}. Furthermore, by shortening and puncturing these codes we are able to obtain more codes with record parameters or missing constructions in \cite{codetables}. For instance, from the code with parameters $[129,41,44]_4$, we obtain the parameters $[129-i-j,41-i,44-j]_4$, for $0\leq i\leq 4$, $0\leq j\leq 3$, which are either records or the construction of a code with those parameters is missing in \cite{codetables}. 
\end{ex}

The next result shows that we can construct quantum codes over $\fq$ using Theorem \ref{thmtraza} together with the Hermitian construction \ref{hermitica}.

\begin{thm}\label{paramquantraza}
Let $\mathcal{A}=\{a_0=0<a_1<a_2<\cdots< a_z\}$ be the set of minimal representatives of the cyclotomic sets $\II_{a_i}$, $0\leq i\leq z$, of $\{0,1,\ldots,q^{2l}-1\}$ with respect to $q^2$. Let $t\leq z$ be an index such that
$$
a_t\leq q^\ell-\left\lfloor \frac{(q-1)}{2}\right\rfloor q^{\ell-1}-\cdots-\left\lfloor\frac{(q-1)}{2}\right\rfloor q-1.
$$
Then, for $\Delta=\bigcup_{i=0}^{t-1} \II_{a_i}\cup \{a_t\}$ as before, assuming that $\Drtshn$ is not degenerate, we have that:
$$
\dim \left( \Drtts\cap \Drtshn\right)\geq \dim \Drtts - 1.
$$
As a consequence, we can construct an EAQECC with parameters
$$
[[n,\geq n-2\sum_{i=0}^{t}n_{a_i}+c,\geq a_{t}+2;c ]]_q,
$$
where $n=q^{2\ell-1}+1$ and $c\leq 1$. 
\end{thm}
\begin{proof}
Similarly to the proof of Theorem \ref{cuanticohermitico}, we are going to consider the code $C=\Drtts$ for the Hermitian construction \ref{hermitica}. By Theorem \ref{thmtraza} we obtain the bound for the minimum distance, and we also obtain that $\dim \Drtts\leq \sum_{i=0}^t n_{a_i}$, which explains the dimension of the quantum code. The only thing left to prove is the claim about the intersection of $\Drtshn$ with its hermitian dual. 

Under our assumptions, in \cite[Thm. 15]{fernandotrace} it is proved that we have $\RS(\Trr,\Delta'')_{q^2}\subset (\RS(\Trr,\Delta'')_{q^2})^{\perp_h}$ for $\Delta''=\bigcup_{i=0}^{t} \II_{a_i}$. The reasoning from the proof of Theorem \ref{cuanticohermitico} finishes the proof.
\end{proof}

\begin{ex}
We continue with Example \ref{ejemplocerostraza}. For $a_t=10$, we have $a_t+2=12$, and, computing the dimension with Magma \cite{magma}, we obtain a quantum code with parameters $[[129,67,12;1]]_2$ using Theorem \ref{paramquantraza}. In the affine case from \cite{fernandotrace}, the parameters $[[128,65,12;0]]_2$ are obtained. Therefore, we have increased the length by 1 and the dimension by 2, at the expense of increasing the parameter $c$ by 1. Moreover, the codes $[[129,73,11;1]]_2,[[129,67,12;1]]_2$ and $[[129,59,13;1]]_2$ that we can obtain in this way (by changing $a_t$) cannot be deduced using the propagation rules from \cite{grasslhowmuch} with the codes in the table of QECCs from \cite{codetables}.
\end{ex}

In \cite{fernandotrace}, the authors consider what they call \textit{complementary codes}, which are obtained in an analogous way, but evaluating in precisely all the points in $\F_{q^{2\ell}}$ besides the zeroes of $\Trr(x)$. For the projective case, it is easy to see that including the point at infinity in this set corresponds to considering the zero set of 
$$\frac{x_0x_1^{q^{2\ell}}-x_0^{q^{2\ell}}x_1}{\Trr(x_1)^h}.$$
All the results we have given in this section apply to these types of codes as well, but with length $q^{2\ell}-q^{2\ell-1}+1$ (instead of $q^{2\ell-1}+1$).

\section*{Declarations}
\subsection*{Conflict of interest} The authors declare that they have no conflict of interest.


\end{document}